\documentclass[draftcls,11pt,onecolumn,letterpaper]{IEEEtran}
\IEEEoverridecommandlockouts
\usepackage{subcaption}
\usepackage{enumerate}
\usepackage{float}
\usepackage{color}
\usepackage{graphicx}
\usepackage{caption}
\usepackage{amsmath} 
\usepackage{amssymb,algorithm}  
\newtheorem{Theorem}{Theorem}[section]

\newcommand{\beq}{\begin{equation}}
\newcommand{\eeq}{\end{equation}}

\newcommand{\bq}[1]{\begin{equation} \label{#1}}
\newcommand{\eq}{\end{equation}}
\newcommand{\bed}{\begin{displaymath}}
\newcommand{\eed}{\end{displaymath}}
\newcommand{\bea}{\bed\begin{array}{rl}}
\newcommand{\eea}{\end{array}\eed}

\newcommand{\barray}{\begin{array}{ll}}
\newcommand{\earray}{\end{array}}
\newcommand{\argmin}{\operatornamewithlimits{argmin}}

\renewcommand{\hat}{\widehat}

\renewcommand{\bar}{\overline}

\newtheorem{Lemma}{Lemma}[section]

\title{Removal of Data Incest in Multi-agent Social Learning in Social Networks }

\author{Maziyar Hamdi$^*$, \textit{Student, IEEE}, Vikram Krishnamurthy, \textit{Fellow, IEEE} \thanks{Maziyar Hamdi
and  Vikram Krishnamurthy are with the Department of Electrical and Computer
Engineering, University of British Columbia, Vancouver, Canada. Email:
{\small \{maziyarh, vikramk\}@ece.ubc.ca.}}
}

\begin{document}

\maketitle
\thispagestyle{empty}
\pagestyle{empty}

\begin{abstract}

\hspace{1mm}Motivated by online reputation systems, we investigate  social learning in a network where agents interact on a time dependent graph to estimate an underlying state of nature. Agents record their own private observations, then update their private beliefs about the state of nature using Bayes' rule. Based on their belief, each agent then chooses an action (rating)  from a finite set and transmits this action over the social network.  An important consequence of such social learning over a network is the ruinous multiple re-use of information known as data incest (or mis-information propagation). In this paper, the data incest management problem in social learning context is formulated on a directed acyclic graph. We  give necessary and sufficient conditions on the graph topology of social interactions to eliminate  data incest. A data incest removal algorithm is proposed such that the public belief of social learning (and hence the actions of agents) is not affected by data incest propagation. This results in an  online reputation system with a higher trust rating.  Numerical examples are provided to illustrate the performance of the proposed optimal data incest removal algorithm.

\end{abstract}

\section{Introduction}\label{sec:intro}

In social learning, agents aim to estimate the state of nature using their private observations and actions from other agents \cite{SL}. The process of updating belief by agents can be done using Bayesian models \cite{Bayesian,Gale} or non-Bayesian models \cite{word,thumb}. In this paper, we consider Bayesian social learning models along with ``\textit{data incest}" also known as mis-information propagation \cite{jstp}. This results in a non-standard information pattern for Bayesian estimation. Before proceeding to the formal definition of data incest in learning over social networks, let us describe the social learning model.

\subsection{Bayesian Social Learning Protocol on Network}\label{subsec:pf}
Consider a social network comprising of $S$ agents that aim to estimate (localize) an underlying state of nature (a random variable). Let $x \in \{\bar x_1,\bar x_2,\cdots,\bar x_X\}$ represent a state of nature (such as quality of a hotel) with known prior distribution $\pi_0$ where $X$ denotes the dimension of the state space. Let $k = 1,2,3,\ldots$ depict epochs at which events occur. These events comprise of taking observations, evaluating beliefs and choosing actions as  described below. The index $k$ depicts the historical order of events and not necessarily absolute time. However, for simplicity, we refer to $k$ as ``time" in the rest of the paper.

The agents use the following Bayesian social learning protocol to estimate the state of nature:
\paragraph*{Step 1. Private observations} To estimate the state of nature $x$, each agent records its $M$-dimensional private observation vector. At each time $k = 1, 2, 3, \ldots$, each agent $s$ ($1 \leq s \leq S$) obtains a noisy private observation $z_{[s,k]}$ from the finite set\footnote{The results of this paper also apply to continuous-valued observations. We consider discrete-valued observations since humans typically record discrete observations.}  $\mathbf{Z} = \{\bar z_1,\bar z_2,\ldots,\bar z_Z\}$ with conditional probability
\begin{equation}\label{eq:B}
B_{ij} = p(z_{[ s,k]} = \bar z_j|x = \bar x_i).
\end{equation} It is assumed that the observations  $z_{[ s,k]}$ are independent random variables with respect to agent $s$ and time $k$\footnote{It is not necessary for agents to record observations at each time $k$ and this does not interfere with the common knowledge assumption in social learning where agents all know about the structure of social learning model. Agents at different time instants  are treated as different nodes in our graphical model. The assumption that  agents record observation at each time $k$  simplifies notation. }.
\paragraph*{Step 2. Private belief} After obtaining its private observation, each agent combines its private observation with the information received form the network to evaluate its belief about state of nature. Define
\begin{eqnarray}\label{eq:sigma}
\Theta_{[ s,k ]} &=& \text{ set of all information received from network  available at agent $s$ at time $k$},\nonumber\\
\end{eqnarray}
 Each agent $s$ combines its  private observation $z_{[ s,k]}$ with the most updated public belief (posterior distribution of state of nature given actions of previous agents)  and evaluates its private belief about state of nature\footnote{The scenario where agents choose their actions according to the most recent public belief is similar to the classical social learning formulation \cite{herd} where actions are transmitted over the network. In this scenario, the information received from the network, $\Theta_{[s,k]}$, is the most updated public belief which is computed  in Step 5. }. Agents use Bayesian social learning to update their private beliefs. The private belief, $\mu_{[ s,k ]}$, is defined as the posterior distribution of the state of nature given the private observation and all information received from other agents in the social network, that is
\begin{equation}\label{eq:privateb1}
\mu_{[ s,k ]} = (\mu_{[ s,k ]}(i), 1\leq i\leq X),\text{where}\quad  \mu_{[ s,k ]}(i) = p\left(x =\bar x_i|\Theta_{[s,k]}, z_{[s,k]}\right).
\end{equation}
Note that the agent's private belief (private opinion) is  not available to other agents or network administrator (who runs the online reputation system).

\paragraph*{Step 3. Local action} Based on its private belief $\mu_{[ s,k]}$, agent $s$ at time $k$ chooses an action $a_{[s,k]}$ from a finite set $ \mathbf{A} = \{1,2,\ldots,A\}$ to minimize its expected cost function (based on the current information available on the network). That is
\begin{equation}\label{eq:action}
a_{[s,k]} = \argmin_{a\in \mathbb{A}}\mathbf{E}\{C(x,a)|\Theta_{[s,k]}, z_{[s,k]}\}.
\end{equation}
Here $\mathbf{E}$ denotes expectation and  $C(x,a)$ denotes the cost incurred by the agent if action $a$ is chosen when the state of nature is $x$.
\paragraph*{Step 4. Social network} Individual agents  broadcast their actions $a_{[s,k]}$ over the social network\footnote{We assume that multiple agents can transmit simultaneously over the network without interfering with each other. This is realistic in a social network, since the time required to exchange  (broadcast) information is substantially smaller than the time to record observations, update beliefs or take actions.}. These actions are  observed by other agents after a random delay. We model this information exchange using a family of directed acyclic graphs. Let\begin{equation}\label{eq:defG} G_{[ s,k]} = (V_{[ s,k]}, E_{[ s,k]}), \quad k = 1,2,3,\ldots, s = 1,2,\ldots,S, \end{equation} denote a sequence of time-dependent graphs of information flow in the social network until and including time $k$. Each vertex in $V_{[ s,k]}$ represents an agent $s$ in the social network at time $k$ and each edge $((s',k'),(s'',k''))$ in $E_{[ s,k]}\subseteq V_{[ s,k]} \times V_{[ s,k]}$ shows that the information (action) of agent $s'$ at time $k'$ reaches agent $s''$ at time $k''$.

\paragraph*{Step 5. Network's public belief update} The network administrator only has access to the local actions of agents. The public belief of the network administrator is the posterior distribution of state of nature given the actions of all agents at previous times. The goal of the paper is to design the network administrator's public belief update. 

As we will see shortly, a major issue with the above protocol is the inadvertent reuse of information (actions of previous agents) which makes the estimates of state of nature biased; that is data incest.

{\bf Aim}: The objective of the paper is to design Step 5 of the above protocol so that when agents use Step 2, data incest is mitigated. The aim is to combine the information received from  agents to compute $p(x|\Theta_{[s,k]})$\footnote{An alternative method is to modify Step 4 such that the network administrator maps the agents' action in Step 3 to a new action and transmits it over the network. In other words, each agent submits a private recommendation (action) and the network administrant suitably modifies this action to avoid incest (using the algorithms we present in Sec.\ref{sec:optimal}).}.

The above protocol models the interaction of agents in a social network that aim to estimate the underlying state of nature $x$. An example is where users aim  to localize a target event by tweeting the location of the detected ``target" on Twitter \cite{earthquake}. Another example is where the state of nature is the true quality of a social unit (such as restaurant). Online reputation systems such as Yelp or Tripadvisor  maintain logs of votes by agents (customers). Each agent visits a restaurant based on reviews on a reputation website such as Yelp. The agent then obtains  private noisy measurement of the state (quality of food in a restaurant). The agent then reviews the restaurant on that reputation website. Such a review typically is a quantized version (for example, rating) of the total information (private belief) gathered by the agent\footnote{The dimension of private beliefs is typically larger than that of actions. Also, individuals tend not to provide their private beliefs at the time of their further social interactions. Therefore, agents map their beliefs to a finite set of actions which are easier to broadcast.}.  With such a protocol, how can agents obtain a fair (unbiased) estimate of the underlying state?\footnote{Having fair estimates of quality of a social unit is a problem of much interest in business. Most of hotel managers (81\%) regularly check the reviews on Tripadvisor \cite{monopoly}.  In \cite{HBR}, it is found that a one-star increase in the average rating of users in Yelp is mapped to about 5-9 \% revenue increase.}. The aim of this paper is for the network administrator to maintain an unbiased reputation system, or alternatively modify the actions of agents, to avoid incest.
\subsection{Context: Data Incest in Social Learning}\label{subsec:context}
\begin{figure}[t]
\centering
\hspace{0cm}\scalebox{.3}{\includegraphics{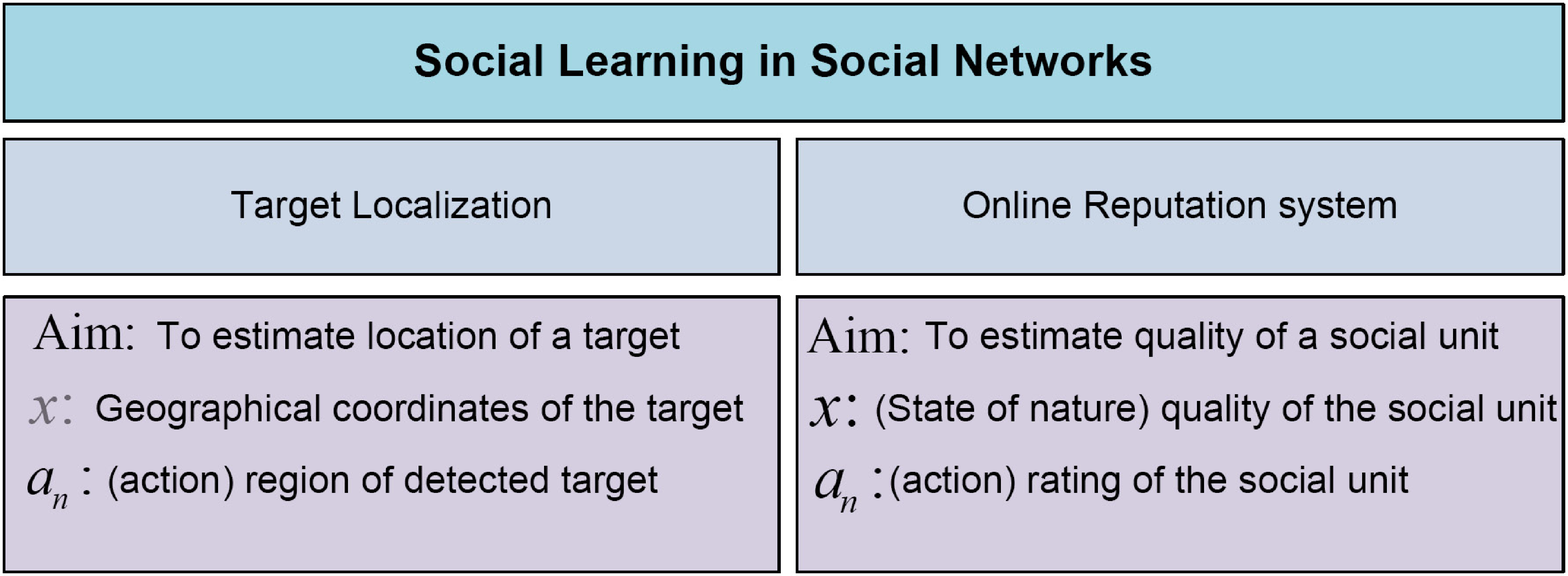}}
\caption{Two examples of multi-agent social learning in social networks: (i) target localization, and (ii) online reputation systems. }
\label{BD4}
\end{figure}
From a statistical signal processing point of view, estimating the state of nature $x$ using the above five-step protocol is non-standard in two ways: First, agents are influenced by the rating of other agents, this is prior influences their posterior and hence their rating.
This effect of agents learning from the actions (ratings) of other agents along with its own private observation is termed ``social learning'' in the economics literature. Social learning can result in interesting phenomenon where rational agents can all end up making the same decision (herding and information cascades; \cite{herd}). Second, (and this effect is more complex), an agent might be influenced by his own rating leading to data incest.

To explain what can go wrong with the above protocol, suppose an agent wrote  a poor rating of the restaurant on a social media site at time 1. Another agent is influenced by this rating and also gives the restaurant a poor rating at time 2. Assume that the diffusion of action is modeled by the graph depicted in Fig.\ref{sample}.  The first agent visits the social media site at time 3 and sees that another agent has also given the restaurant a poor rating - this double confirms his rating and he enters another poor rating.  In a fair system, the first agent should have been aware that the rating of the second agent was influenced by his rating - so that first agent has effectively double counted his first rating by casting the second poor rating. Data incest is a consequence of the recursive nature of Bayesian social learning and the communication graph. The data incest in a social network is defined as the naive re-use of actions of other agents in the formation of the belief of an agent when these actions could have been initiated by the agent.

\begin{figure}[h]
\centering
\hspace{0cm}\scalebox{.16}{\includegraphics{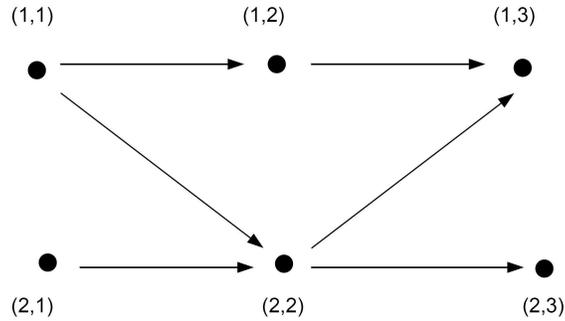}}
\caption{ \label{sample} Example of communication graph, with two agents ($S = 2$) and over three event epoches ($K = 3$). The arrows represent exchange of information regarding actions taken by agents.}
\end{figure}

The two effects of social learning and data incest lead to non-standard information patterns in  state estimation. Herd occurs when the public belief overrides the private observations and thus actions of agents are independent of their private observations. An extreme case of this is an information cascade when  the public belief of social learning hits a fixed point and does not evolve any longer. Each agent in a cascade acts according to the fixed public belief and social learning stops \cite{herd}\footnote{There are subtle differences between an individual agent herding, a herd of agents and an information cascade; see for example \cite{herd,KP13}.}.
Data incest results in bias in the public belief as a consequence of the unintentional re-use of identical actions in the formation of public belief of social learning; the information gathered by each agent is mistakenly considered to be independent. This results in over confidence and bias in estimates of state of nature. Due to the lack of information about the  topology of the communication graph, data incest arises in Bayesian social learning in social networks. Therefore, the Bayesian social learning protocol requires a careful design to ensure that data incest is mitigated. The aim of this paper is to modify the five-step protocol presented in Sec.\ref{subsec:pf} such that data incest does not arise. As we will see in Sec.\ref{subsec:dis}, the proposed data incest removal algorithm can be applied to the state estimation problems shown in Fig.\ref{BD4}.

\subsection{Main Results and Paper Organization:}
With the above five-step social learning protocol in social networks, we are now ready to outline the main results of this paper:
\begin{enumerate}
\item  In Sec.\ref{sec:social}, the data incest problem is formulated on a family of time dependent directed acyclic graphs
\item In Sec.\ref{sec:optimal}, a necessary and sufficient condition on the graph is  provided for exact data incest removal. This constraint is on the topology of  communication delays (communication graph). Also examples where exact incest removal is not possible are illustrated.
\item A data incest removal algorithm is proposed for the five-step social learning protocol in Sec.\ref{sec:optimal}. The data incest removal algorithm is employed by the network administrator to update the public belief in Step 5 of the social learning protocol of Sec.\ref{subsec:pf}\footnote{In this paper we consider Bayesian estimation over a finite time horizon.  We do not consider  the asymptotic agreement of social learning or consensus formation in social networks. Consensus formation is asymptotic and typically non-Bayesian. From a practical point of view, information exchange in a social network is typically over a finite horizon.}.

\end{enumerate}

Finally in Sec.\ref{sec:num}, numerical examples are provided which illustrate the data incest removal algorithm.
\subsection{Related Works:}
Social learning theory is used to investigate the learning behavior of agents in social and economic networks \cite{SL}. There are several papers in the literature discussing Bayesian models \cite{BN,Bayesian,AceSL,ghana,KP13} and non-Bayesian models \cite{word,thumb} for social learning. Different models for diffusion of beliefs in social networks are presented in \cite{diffusion}. For a comprehensive survey on herding and information cascade in social learning, see \cite{herd}. Stochastic control with social learning for sequential change detection problems is considered in \cite{K12}.

 Mis-information in the signal processing literature refers to faulty or inaccurate information which is broadcasted unintentionally.  A different type of mis-information called ``gossip'' is investigated in \cite{gossipp} where non-Bayesian models are employed. A model of Bayesian social learning where agents receive private information about state of nature and observe the actions of their neighbors is investigated in \cite{tamuz}. They proposed an algorithm for agents' calculations on tree-based social networks and analyzed their algorithm in terms of efficiency and convergence. Another category of mis-information caused by influential agents (agents who heavily affect  actions of other agents in social networks) is investigated in \cite{SL}. Mis-information in the context of this paper is motivated by sensor networks where the term ``data incest" is used \cite{vikram2}. In multi-agent social learning in networks, data incest occurs when information (action) of one agent is double-counted by other agents (due to the lack of information about the topology of the communication graph); this yields to overconfidence.

 The overconfidence phenomena (caused  by data incest) also arises in  Belief Propagation (BP) algorithms \cite{Pearl, Murphy} which are used in computer vision and error-correcting coding theory. The aim of BP algorithms is to solve inference problems over graphical models such as Bayesian networks (where nodes represent random variables and edges depict dependencies among them) by computing a marginal distribution. BP algorithms require passing local messages over the graph (Bayesian network) at each iteration. These algorithms converge to the exact marginal distribution when the factor graph is a tree (loop free). But for graphical models with loops, BP algorithms are only approximate due to the over-counting of local messages \cite{Weiss} (which is similar to data incest in multi-agent social learning)\footnote{ There exists some similarities between BP and social learning in the sense that they are both systematic structures to perform Bayesian inference over graphs. However, they are not related in principle. While graphs represent social interactions among agents in social learning, graphical models in BP depict the conditional dependency between nodes (random variables)-- they do not imply the actual communications, see\cite{diffusion}.}.  With the algorithm presented in Sec.\ref{sec:optimal}, data incest can be mitigated from Bayesian social learning over non-tree graphs that satisfy a topological constraint.

The closest work to the current paper is \cite{jstp}. In \cite{jstp,fusion}, data incest is considered in a network where agents exchange their private belief states - that is, no social learning is considered. In a social network, agents rarely exchange private beliefs, they typically broadcast actions (votes) over the network. Motivated by trustable online reputation systems, we consider data incest in a social learning context with social network structure where actions (or equivalently public belief of the social learning) are transmitted over the network. This is quite different from private belief propagation in social networks. Simpler versions of this information exchange process and estimation were investigated by Aumann \cite{aumann} and Geanakoplosand and Polemarchakis \cite{GP}. The results derived in this paper extend theirs.

Finally, the methodology of this paper can be interpreted in terms of the recent Time magazine article \cite{timem} which provides interesting rules for online reputation systems. These include: (i) review the reviewers, (ii) censor fake (malicious) reviewers. The data incest removal algorithm proposed in this paper can be viewed as  ``reviewing the reviews" of other agents to see if they are associated with data incest or not.

\subsection{Limitations} We do not consider the case where the network is not known to the administrator. The state of nature in this paper is a random variable and we do not allow for estimating a random process. In this paper, we consider Bayesian estimation over a finite time horizon and the asymptotic agreement of social learning is not considered in this paper.
\section{Social Learning Over Social Networks}\label{sec:social}
 In this section, a graphical model is presented for the five-step social learning protocol introduced in Sec.\ref{subsec:pf}. In the evaluation of the private belief by agents, data incest may arise as a result of abusive re-use of information of other agents (caused by the lack of information about the topology of the network and the recursive nature of Bayesian models).

\subsection{Social Network Communication Model}\label{subsec:inf}
 With the five-step social learning protocol presented in Sec.\ref{subsec:pf} and the graph theoretic definitions provided in Appendix~\ref{subsec:graphtheory}, here we discuss the diffusion of actions in the resulting social network. For notational simplicity, instead of $[s,k]$, the following scalar index $n$  is used:
\begin{equation} \label{reindexing_scheme}
 n \triangleq s+ S(k-1), \quad
s \in \{1,\ldots, S\}, \; k \in \{1,2,3,\ldots \}\;.
\end{equation}
Recall that in the social learning model considered in the paper, the historical order of events is important and  $k$ is used to denote the order of occurrence of events in real time.  Subsequently, we will refer to $n$ as a ``node" of the time dependent communication graph $G_n$. Recall from Sec.\ref{sec:intro}, $G_n = (V_n, E_n)$ denotes the time-dependent communication graph of the social network. Each node $n'$ in $G_n$ represent an agent $s'$ at time $k'$ such that $n' = s' + S(k'-1)$, see (\ref{reindexing_scheme}). Each directed edge of $G_n$ shows a communication link in the social network represented by $G_n$. This means that if $(n,n') \in E_n$, agent $s'$ at time $k'$ uses the information of agent $s$ at time $k$ to update his private belief about the underlying state of nature $x$. Note that with the way we defined the communication graph, $G_n$ is always a sub-graph of $G_{n+1}$. Therefore, as the following theorem proves, diffusion of actions can be modeled via a family of time-dependent Directed Acyclic Graphs (DAGs)\footnote{See (\ref{eq:transitivieclosurematrix}) in Appendix~\ref{subsec:graphtheory}.}.
\begin{Theorem} \label{Theorem:informationflowDAG}
\it The information flow in a social learning over social networks comprising of $S$ agents for $k = 1,2,3,\ldots, K$ can be represented by a family of DAGs ${\mathcal G} = \{G_n\}_{n\in \{1,\ldots,N\}}$ where $N=SK$. Each DAG $G_n=(V_n,E_n)$ represents the information flow between the $n$ first nodes, where the generic node $n$ is defined by (\ref{reindexing_scheme}).
\end{Theorem}
\begin{proof}
The proof is presented in Appendix \ref{subsec:proofp1}.
\end{proof}
The adjacency and the transitive closure matrices of $G_n$ are denoted by $\bar A_n$ and $T_n$, respectively (see Appendix~\ref{subsec:graphtheory}).
Using the adjacency and transitive closure matrices of $G_n$, all nodes whose information are available at node $n$ can be found. From that, $\Theta_n$ is specified (see Eq. (\ref{eq:sigma})).
\subsection{Constrained Social Learning in  Social Networks}
The observation process and the evaluation of private belief are described in Sec.\ref{subsec:pf}. With the scalar index defined in (\ref{reindexing_scheme}), the observation vector and the private belief  of  node $n$ (that represents agent $s$ at time $k$) are denoted by $z_n$ and $\mu_n$, respectively.  Recall from Sec.\ref{sec:intro} that the public belief of the network is the posterior distribution of state of nature given all information available at node $n$, that is\footnote{Recall from Sec.\ref{subsec:pf} that $\Theta_n$ denotes the set of all information available at node $n$.}
\begin{equation}\label{eq:public}
\pi_{-n} = (\pi_{-n}(i), 1\leq i \leq X),\text{where} \quad \pi_{-n}(i) = p(x = \bar x_i|\Theta_{n}).
\end{equation}
After action $a_n$ is chosen (by agent $s$ at time $k$), the public belief of social learning changes (because of the most recent action $a_n$). To avoid confusion, we define ``\textit{after-action}" public belief which includes action $a_n$, that is
\begin{equation}\label{eq:public}
\pi_{+n} = (\pi_{+n}(i), 1\leq i \leq X),\text{where} \quad \pi_{+n}(i) = p(x = \bar x_i|\Theta_n, a_n).
\end{equation}

Note that the scenario where agents in social network have access to the most up-to-date public belief and updates their private beliefs accordingly (using the public belief of social learning and their private observations) is similar to the classical social learning setup where agents transmit their actions over the network. As discussed in Sec.\ref{sec:intro}, here instead of classical social learning setup, we consider a scenario that the network administrator evaluates the after-action public beliefs of agents and transmits them over the network.
When action $a_n$ is chosen by an agent and submitted to the network, the network administrator computes the corresponding after-action public belief immediately (without delay) and broadcasts it over the network. As discussed in Sec.\ref{sec:intro}, due to the communication delay, the transmitted belief reaches other agents after a random delay (which is modeled via communication graph $G_n$ with adjacency matrix $\bar A_n$ and transitive closure matrix $T_n$). In this scenario  $\Theta_n$ is defined as
\begin{equation}\label{eq:theta}
\Theta_{n} = \{\pi_{+i}; \quad  (1\leq i\leq n-1) \text{ and } \bar A_n(i,n) = 1 \}
\end{equation}    The following lemma shows how to update private belief, choose local action , and finally update the after-action public belief using $\Theta_n$ and $z_n$.
\begin{Lemma}\label{lem:1}
Consider the five-step  social learning protocol presented in Sec.\ref{subsec:pf} with $S$ agents and the communication graph $G_n$. Let $\pi_{-n}$ denote the public belief of social network at this node. Then, the social learning elements (private belief, action, and after-action public belief) of node $n$ with observation vector $z_n = \bar z_l$ can be computed from ($1 \leq m \leq X$)
\begin{align}\label{eq:pb}
&\mu_n(m) = p\left(x = \bar x_m|\Theta_n, z_n\right) \propto c\pi_{-n}(m)B_{ml},  \nonumber\\
&a_n = \argmin_{a\in \mathbb{A}}\mathbf{E}\{C(x,a)|\Theta_n, z_n\}= \argmin_{a\in \mathbb{A}}\mathbf{E}\{C_a'\mu_n\}, \nonumber\\
&\pi_{+n}(m)\propto c \pi_{-n}(m)\sum_{j=1}^{Z}\left[\prod_{\hat a \in \mathbf{A}-\{a_n\}}\mathbb{I}(C_{a_n}'B_{j}\pi_{-n} <C_{\hat a}'B_{j}\pi_{-n} )\right]B_{mj},\\\nonumber
\end{align}
where $c$ is a generic normalizing constant, $B_j = {\rm diag}(B_1j,\ldots, B_{Xj})$, and $\mathbb{I}(\cdot)$ is indicator function. Here, $C_a$ is the cost vector defined as $C_a = [C(1,a) ~~C(2,a)~~ \ldots ~~C(X,a)]$.

\end{Lemma}
\begin{proof}
The proof is presented in Appendix \ref{ap:lem1}.
\end{proof}
Lemma~\ref{lem:1} summarizes the social learning problem considered in this paper. Node $n$ receives a set of after-action public beliefs and  combines them to compute $\pi_{-n}$ from which the private belief can be computed (using private observation $z_n$).  As described in Sec.\ref{sec:intro}, a major issue with the above protocol is data incest (because information in $\Theta_n$ are mistakenly considered to be independent in information aggregation process). The aim of this paper is to devise a data incest removal algorithm for the network administrator to deploy such that the estimates of agents are unbiased.

To formulate the data incest problem that arises in the five-step protocol of Sec.\ref{sec:intro}, two following types of social learning protocols are presented: (i) constrained social learning protocol which we introduce shortly (see Fig.\ref{BD1}), and (ii) idealized social learning which is investigated in Sec.\ref{sec:optimal}.
\begin{figure}[t]
\centering
\hspace{-1.2cm}\scalebox{.27}{\includegraphics{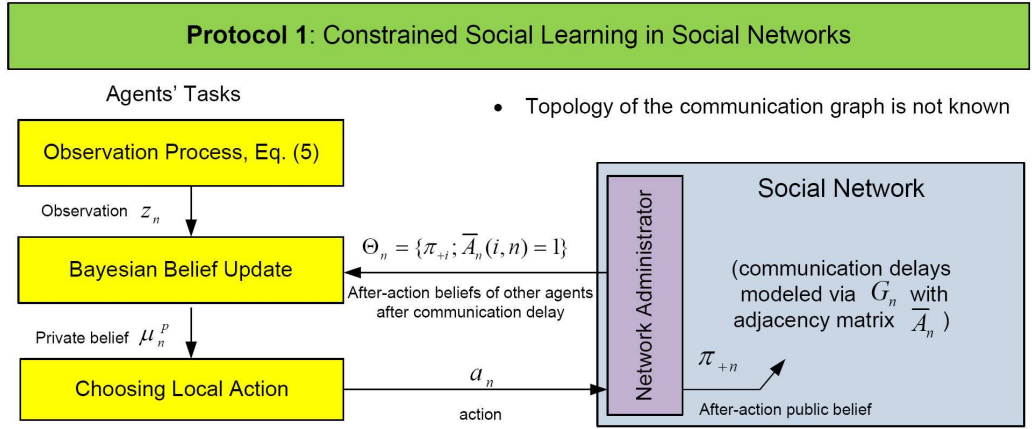}}
\caption{Protocol~1: Constrained social learning in social networks described in Sec.\ref{subsec:pf}. As a result of random (unknown) communication delays, data incest arises.}
\label{BD1}
\end{figure}
 The five-step constrained social learning protocol introduced in Sec.\ref{subsec:pf}, is illustrated in Fig.\ref{BD1}.
 Note that in the constrained social learning problem described in Protocol~1, agents do not have information about the communication graph. This is why the term ``constrained" is used.  With Protocol~1, the constrained social learning in social networks can be summarized as
\begin{eqnarray}\label{eq:setup1}&&\hspace{-.8cm}\left\{\begin{array}{l}
z_{n} \sim  B_{iz} ,\quad x = i, \quad \text{{ (observation process)}}\\
a_{n} = \mathcal{S}(\Theta_{n},z_{n}), \quad \text{(choosing local action)}\\
\pi_{+n} = \mathcal{A}(\Theta_n,a_n) \quad \text{(updating the after action belief)}
\end{array}\right.
\end{eqnarray}
In (\ref{eq:setup1}) Algorithm $\mathcal{A}$ is employed by the network administrator to update the after-action public belief. Due to the lack of knowledge about the communication graph  (and the recursive nature of Bayesian models), data incest arises in constrained social learning if algorithm $\mathcal{A}$ is not designed properly. Algorithm $\mathcal{S}$ in (\ref{eq:setup1}) is employed by each agent to choose its action; this algorithm can be constructed using the results of Lemma~\ref{lem:1}. The aim of this paper is to devise the algorithms $\mathcal{A}$ and $\mathcal{S}$ such that the public belief of social learning (or equivalently actions $a_n$ for all $n = 1,2,\ldots$) are not affected by data incest. 

\textbf{Remark 1:} In order to choose an action from the finite set of all possible actions, agents minimize a cost function. This cost function can be interpreted in terms of the reputation of agents in online reputation systems. For example if the quality of a restaurant is good and an agent wrote a bad review for it in Yelp and he continues to do so for other restaurants, his reputation becomes lower among the users of Yelp. Consequently, other people ignore reviews of that (low-reputation) agent in evaluation of their opinion about the social unit under study (restaurant). Therefore, agents  minimize the penalty of writing inaccurate reviews (or equivalently increase their reputations) by choosing proper actions. This behavior is modeled by minimizing a cost function in our social learning model.

\textbf{Remark 2:} In comparison to the public belief which can be computed by the network administrator (who monitors the agents' actions and communication graph), the agents' private beliefs  cannot be computed by the network administrator. The private belief depends on the local observation which is not available to the network. Note that in Step~2 of the constrained social learning Protocol~1, the results of Lemma~\ref{lem:1} are used to compute $ \mu_n$ using $z_n$ and $\pi_{-n}$.

\section{Data Incest Removal Algorithm}\label{sec:optimal}
So far in this paper, Bayesian social learning model and communication among agents in social networks have been described. This section presents the main result of this paper, namely the solution to the constrained social learning problem (\ref{eq:setup1}). We propose a data incest removal algorithm such that the public belief of social learning (and consequently the chosen action) is not affected by data incest. To devise the data incest removal algorithm, an idealized framework is presented that prevents data incest as we will describe shortly. Comparing the public belief of the idealized framework with the same of the constrained social learning, the data incest removal algorithm is specified. This data incest removal algorithm is used by the network administrator and replaces Step 5 of the social learning protocol presented in Sec.\ref{subsec:pf}. A necessary and sufficient condition for the data incest removal problem is also presented in this section.
\subsection{The Idealized Benchmark for Data Incest Free Social Learning in Social Networks}\label{subsec:bench2}
In this subsection, an idealized (and therefore impractical) framework that will be used as a benchmark to derive the constrained social learning protocol, is described. In the idealized protocol, it is assumed that the entire history of actions  along with the communication graph are known at each node.  Due to the knowledge about the entire history of actions and the communication graph (dependencies among actions) in the idealized framework, data incest does not arise\footnote{In the constrained social learning algorithm, each node receives the most recent after-action public beliefs of its neighbors or equivalently the updated public belief.}. Define \begin{equation}
\Theta^{\rm full}_{n} = \{a_i; \quad  (1\leq i\leq n-1) \text{ and } T_n(i,n) = 1 \}.
\end{equation}
Here, $T_n$ is the transitive closure matrix of the communication graph $G_n$. In the idealized framework, the public belief can be written as \begin{equation}\label{temp}
p(x|\Theta_{n}^{\rm full}) \propto \pi_0 \prod_{a_i \in \Theta_n^{\rm full}} p(a_i|x, S_i),\end{equation} where $S_i \subset \Theta_n^{\rm full}$ denotes the set of actions that $a_i$ depends on them. The public belief in the idealized social learning is free of data incest, as it can be inferred from (\ref{temp}).  The idealized social learning in social networks (Protocol~2) is illustrated in Fig.\ref{BD3}. The private belief of node $n$ in the idealized social learning is denoted by $\hat \mu_n$.

\begin{figure}[t]
\centering
\hspace{-1.2cm}\scalebox{.27}{\includegraphics{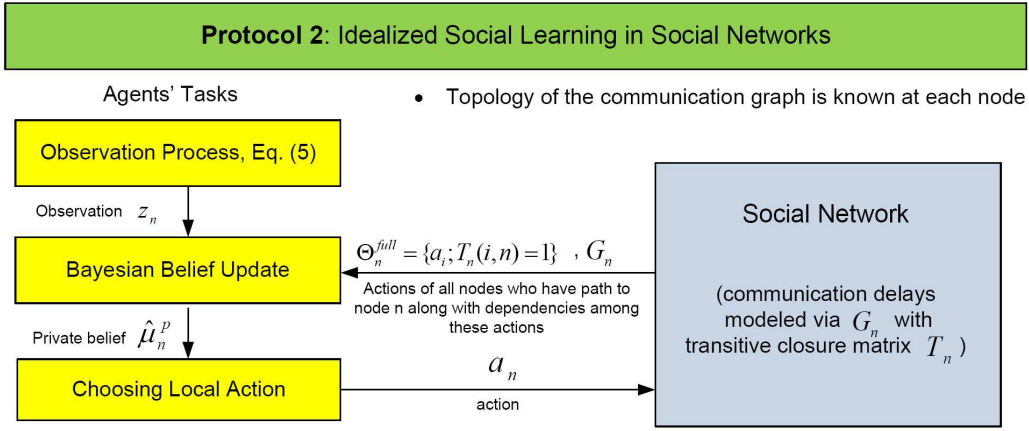}}
\caption{Protocol~2: Idealized benchmark social learning in social networks. In this protocol, the complete history of actions chosen by agents and the communication graph are known. Hence, data incest does not arise. This benchmark protocol will be used to design the data incest removal protocol.}
\label{BD3}
\end{figure}

With Protocol~2, the idealized social learning problem can be summarized as
\begin{eqnarray}\label{eq:setup2}&&\hspace{-.8cm}\left\{\begin{array}{l}
G_{n} = (V_n, E_n) \hspace{2mm}  \text{\small{is given.}}\\
z_{n} \sim  B_{iz} ,\quad x = i, \quad \text{\small{ (observation process)}}\\
a_{n} = \mathcal{F}(\Theta^{\rm full}_{n},z_{n},G_n), \quad \text{(choosing local action)}\\
\hat\pi_{+n} = \mathcal{B}(\Theta^{\rm full}_n,a_n) \quad \text{(updating the after action belief)}
\end{array}\right.
\end{eqnarray}
 We use $\hat \pi_{-n}$ and $\hat \pi_{+n}$ to denote the public belief and the after-action public belief of social learning at node $n$ in the idealized social learning Protocol~2, respectively. Algorithm $\mathcal{B}$ in the social learning problem (\ref{eq:setup2}) is the ordinary recursive Bayesian filter to update $\hat\pi_{+n}$. In Algorithm $\mathcal{B}$, first $\hat\pi_{-n}$ is computed via (\ref{temp}) and then the results of Lemma~\ref{lem:1} are applied to choose $a_n$. Note that if there exists a path between node $i$ and node $n$, then action $a_i\in \Theta_n^{\rm full}$. Since the history of actions and the communication topology are available in the idealized social learning Protocol~2, $\hat\pi^{p}_n$ is free of data incest.

\textbf{Remark 3:} The idealized social learning Protocol~2 requires agents to know the entire history of actions and also the dependencies among these actions which seems impractical in large social networks. However, this protocol serves as a benchmark against the constrained social learning Protocol~1, hence  its practicality is irrelevant. The aim of this idealized benchmark protocol is to specify the data incest removal algorithm which can be done by comparing the social learning public belief of Protocol~1 and Protocol~2 as we will show shortly.

\subsection{The Data Incest Free Belief in the Idealized Social Learning Protocol~2}

The goal of this paper is to replace Step 5 of the five-step constrained social learning Protocol~1 with an algorithm that mitigates data incest. As described earlier, to solve the data incest management problem, we introduced the idealized social learning Protocol~2 that prevents data incest. By comparing the after-action public beliefs of agents in the idealized social learning Protocol~2 with those in the constrained social learning Protocol~1, the data incest removal algorithm can be invented. In this subsection, an expression is derived for the after-action public beliefs of agents in the idealized social learning Protocol~2. Let $\theta_n^{\rm full}$ denote the logarithm of the after-action public belief $\hat \pi_{+n}$  in the idealized social learning Protocol~2, that is
  \begin{equation}\label{eq:deftheta}
  \theta^{\rm full}_n = \log\left(p(x|\Theta_n^{\rm full},a_n,G_n)\right).
  \end{equation}
Theorem~\ref{prop:ideal} below gives an expression for the after-action public belief in the idealized social Protocol~2.
 \begin{Theorem} \label{prop:ideal}
\it Consider problem (\ref{eq:setup2}) with the idealized social learning Protocol~2. The data incest free after-action public belief of node $n$ (which represents agent $s$ at time $k$ according to re-indexing equation (\ref{reindexing_scheme})) is:

\begin{equation}\label{eq:pbideal}
 \theta_{n}^{\rm full} = \sum_{i=1}^{n-1} t_n(i)\nu_i + \nu_n,
\end{equation}
where $\nu_k$ denotes $\log\left(p(a_{k}|x,S_k)\right)$.
 Recall that $t_n$ defined in (\ref{deft}) in Appendix~\ref{subsec:graphtheory} as the first $n-1$ elements of the $n^{th}$ column of $T_n$.
\end{Theorem}
\begin{proof}
The proof is presented in Appendix~\ref{subsec:proofp2}.
\end{proof}
As can be seen in (\ref{eq:pbideal}), the (logarithm of the) after-action public belief of node $n$ can be written as a linear function in terms of $\nu_i$ using $t_n$. Due to this linearity, the data incest removal algorithm can be constructed as we will explain later in this section.  Also (\ref{eq:pbideal}) implies that the optimal data incest free after-action public beliefs of agents in the idealized social learning Protocol~2 depend on the communication graph explicitly in terms of the transitive closure matrix\footnote{See (\ref{eq:transitivieclosurematrix}) in Appendix~\ref{subsec:graphtheory}.}. Basically  the non-zero elements of $t_n$ show all nodes who have a path to node $n$ and thus their actions contribute in the formation of the private belief of node $n$. Eq. (\ref{eq:pbideal}) is quite intuitive from the fact that each agent employs a recursive Bayesian filter to combine its private observation with the information received from the network.
\subsection{Data Incest Removal Algorithm for Problem (\ref{eq:setup1}) With Constrained Social Learning Protocol~1}\label{subsec:optimal}
Given the expression for the after-action public belief of the idealized social learning Protocol~2, the aim here is to propose an optimal information aggregation scheme (that replaces Step~5) such that the after-action public belief of the constrained social learning Protocol~1 is equal to the same of the idealized social learning Protocol~2 (which is free of data incest)\footnote{From that, algorithm $\mathcal{A}$ in problem (\ref{eq:setup1}) with constrained social learning Protocol~1 can be specified.}.   That is,
  \begin{equation}\label{eq:mis1}
  p(x|\Theta_{n},a_{n}) = p(x|\Theta_n^{\rm full},a_n,G_n).
  \end{equation}
Similar to $\theta^{\rm full}_n$, let $\hat \theta_n$ denote the logarithm of the after action public belief of node $n$,
\begin{align}\label{eq:defprop}
&\hat\theta_n = \log\left(p(x|\Theta_{n},a_n)\right).
\end{align}
 We propose the following optimal information aggregation scheme to evaluate the after-action public belief using a $n-1$ dimensional weight vector $w_n$ as follows,
\begin{align}\label{def:constraintestimate}
\hat \theta_{n} =  \sum_{i=1}^{n-1} w_n(i)\hat\theta_i + \nu_{n},
\end{align}
where $w_n$ with elements $w_n(i)$ ($1\leq i\leq n-1$) is defined more precisely in (\ref{eq:wn}). Using optimal information aggregation scheme (\ref{def:constraintestimate}) and (\ref{eq:pb}) in Lemma~\ref{lem:1}, algorithm $\mathcal{A}$ in (\ref{eq:setup1}) can be specified.\\

 \textbf{Remark 4:} The optimal information aggregation scheme (\ref{def:constraintestimate}) is deployed by the network administrator in Step 5 of the social learning protocol presented in Sec.\ref{subsec:pf} to combine the received information (after-action public beliefs or equivalently actions) form other nodes and computes $\sum_{i=1}^{n-1} w_n(i)\hat\theta_i$, this is the public belief of social learning at node $n$. Then, node $n$ updates its private belief based on the most updated public belief (provided by the network administrator) and chooses its action $a_n$ accordingly and then transmits it over the network. Then, the network administrator evaluates $\nu_n$ and updates the after-action public belief by computing $\hat \theta_{n} =  \sum_{i=1}^{n-1} w_n(i)\hat\theta_i + \nu_{n}$. Alternatively, the network administrator can compute the most recent after-action public belief of nodes and transmits it over the network. In this case, node $n$ combines the received after-action public beliefs using the optimal weight vector $w_n$ and chooses its action $a_n$. Then, the action $a_n$ is broadcasted to the network administrator and the after-action public belief of node $n$ is updated accordingly.

 The weight vector $w_n$ depends on the communication graph and can be computed simply by (\ref{eq:wn}). Theorem~\ref{prop:mis} below proves that by using the optimal information aggregation scheme (\ref{def:constraintestimate}) with $w_n$ defined in (\ref{eq:wn}), data incest can be completely mitigated. However, for some network topologies, it is not possible to remove data incest completely.  The following constraint presents the necessary and sufficient condition on the network for the exact data incest removal.\\
\textbf{Topological Constraint 1:} Consider the constrained social learning problem (\ref{eq:setup1}) with Protocol~1. Then, the weight vector $w_n$ used in optimal information aggregation  scheme (\ref{def:constraintestimate}) satisfies the topological constraints if $\forall j\in\{1,\ldots,n-1\}$ and $\forall n\in\{1,\ldots,N\}$ \begin{eqnarray}\label{constraint}
b_n(j)=0   & \implies w_n(j)= 0,
\end{eqnarray}
where $b_n$ is defined in (\ref{deft}) and denotes the $n$-th column of the adjacency matrix of $G_n$. Basically Constraint~1 puts the ``\textit{availability constraint}'' on the communication graph. This means that if information of node $j$ is required at node $n$ ($w_n(j) \neq 0$), there should be a communication link between node $j$ and node $n$ ($b_n(j) \neq 0$).  Assuming that Constraint 1 holds, Theorem~\ref{prop:mis} below ensures that the (after-action) public belief of nodes in problem  (\ref{eq:setup1}) with the constrained social learning Protocol~1 is identical to the same of the problem (\ref{eq:setup2}) with the idealized social learning Protocol~2.
\begin{Theorem}\label{prop:mis}
 \it Consider problem (\ref{eq:setup1}) with the constrained social learning Protocol~1 of Sec.\ref{sec:social}. Then using the optimal information aggregation scheme (\ref{def:constraintestimate}), data incest can be mitigated by using the optimal set of weights $\{w_n\}_{n\in\{1,\ldots,N\}}$ given that the topological Constraint~1 is satisfied. The optimal weight vector is
\begin{equation}\label{eq:wn}
w_n =  t_n\left(\left(T_{n-1}\right)'\right)^{-1}.
\end{equation}
By using the optimal combination scheme (\ref{def:constraintestimate}) and optimal weight vector defined in (\ref{eq:wn}), the data incest in social learning problem (\ref{eq:setup1}) is completely mitigated, that is $\hat{\theta}_{n}  =  \theta_{n}^{\rm full}$ if $w_n$ satisfies topological Constraint~1 where $\hat{\theta}_n$ and $\theta_n^{\rm full}$ are defined in (\ref{eq:defprop}) and (\ref{eq:deftheta}) respectively. Recall that $t_n$ is defined in (\ref{deft}) as the first $n-1$ elements of the $n^{th}$ column of $T_n$.
\end{Theorem}
\begin{proof}
The proof is presented in Appendix~\ref{subsec:proofp3}.
\end{proof}
In the proposed optimal information aggregation scheme (\ref{def:constraintestimate}), the after-action public belief of each node (which represents an agent at specific time instant) can be written as a linear combination of the received information from the network and the local action (combined with the optimal weight vector $w_n$).  This is quite intuitive from the linearity of the after-action public belief in the idealized social learning Protocol~2. Theorem~\ref{prop:ideal} derives the after-action public belief of agents in problem (\ref{eq:setup2}) with the idealized social learning Protocol~2, $\theta_n^{\rm full}$, as a linear function of received information from the social network. Then, assuming that $\hat\theta_{1:n-1} = \theta_{1:n-1}^{\rm full}$, the optimal weight vector $w_n$ is specified such that $\hat\theta_n =\theta^{\rm full}_n$. Theorem~\ref{prop:mis} proves that if $w_n =t_n\left(\left(T_{n-1}\right)'\right)^{-1}$ then $\hat\theta_n =\theta^{\rm full}_n$.

Using the optimal information aggregation scheme (\ref{def:constraintestimate}), the five-step Bayesian social learning protocol in Sec.\ref{subsec:pf} with data incest removal algorithm can be summarized as
\begin{algorithm}[H]\floatname{algorithm}{Algorithm} \begin{minipage}{17cm}\footnotesize \vspace{0mm} {\sf   Step 1. Observation process}: Private observation vector $z_n$ is obtained according to (\ref{eq:B}).\\
{\sf Step 2. Private belief}: Node $n$ accesses the network and evaluates its private belief according to (\ref{eq:privateb1}) using the most updated public belief $\pi_{-n}$. \\
{\sf Step 3. Local action}: Action $a_n$ is chosen via (\ref{eq:action}).\\
{\sf Step 4. Social network}: The network administrator evaluates (maps $a_n$ to) the after-action public belief using the information aggregation scheme (\ref{def:constraintestimate}) and transmits it over the network.\\
{\sf Step 5. Public belief update}: The network administrator provides the optimal weight vector $w_n$ to nodes to combine the information received from the network $\Theta_n$ or alternatively it can provides the most up-to-date public belief $\pi_{-n}$ to each node.
\caption{Constrained Bayesian social learning with data incest removal algorithm at each node $n$}\label{protocol:cons}
\end{minipage}
\end{algorithm}


\textbf{Discussion of topological constraint (\ref{constraint}):} The non-zero elements of $w_n$ depict the nodes whose information are required at node $n$ to remove data incest. This puts a topological constraint on the communication graph. If $w_n(j)$ is non-zero, this means that information of node $j$ is needed at node $n$ and there should be an edge in $G_n$ that connects node $j$ to node $n$, this is the topological Constraint~1. Constraint~1 ensures that the essential elements for data incest removal are available at  node $n$ and Theorem~\ref{prop:mis} specifies the exact data incest removal algorithm. From Theorem~\ref{prop:mis}, it is simple to show that Constraint~1 is a necessary and sufficient condition for data incest removal in learning problem (\ref{eq:setup1}). Consider two examples of communication graph shown in Fig.\ref{tp}.
\begin{figure}[h]
\centering
\begin{minipage}[b]{.45\textwidth}
\hspace{-.5cm}\scalebox{.3}{\includegraphics{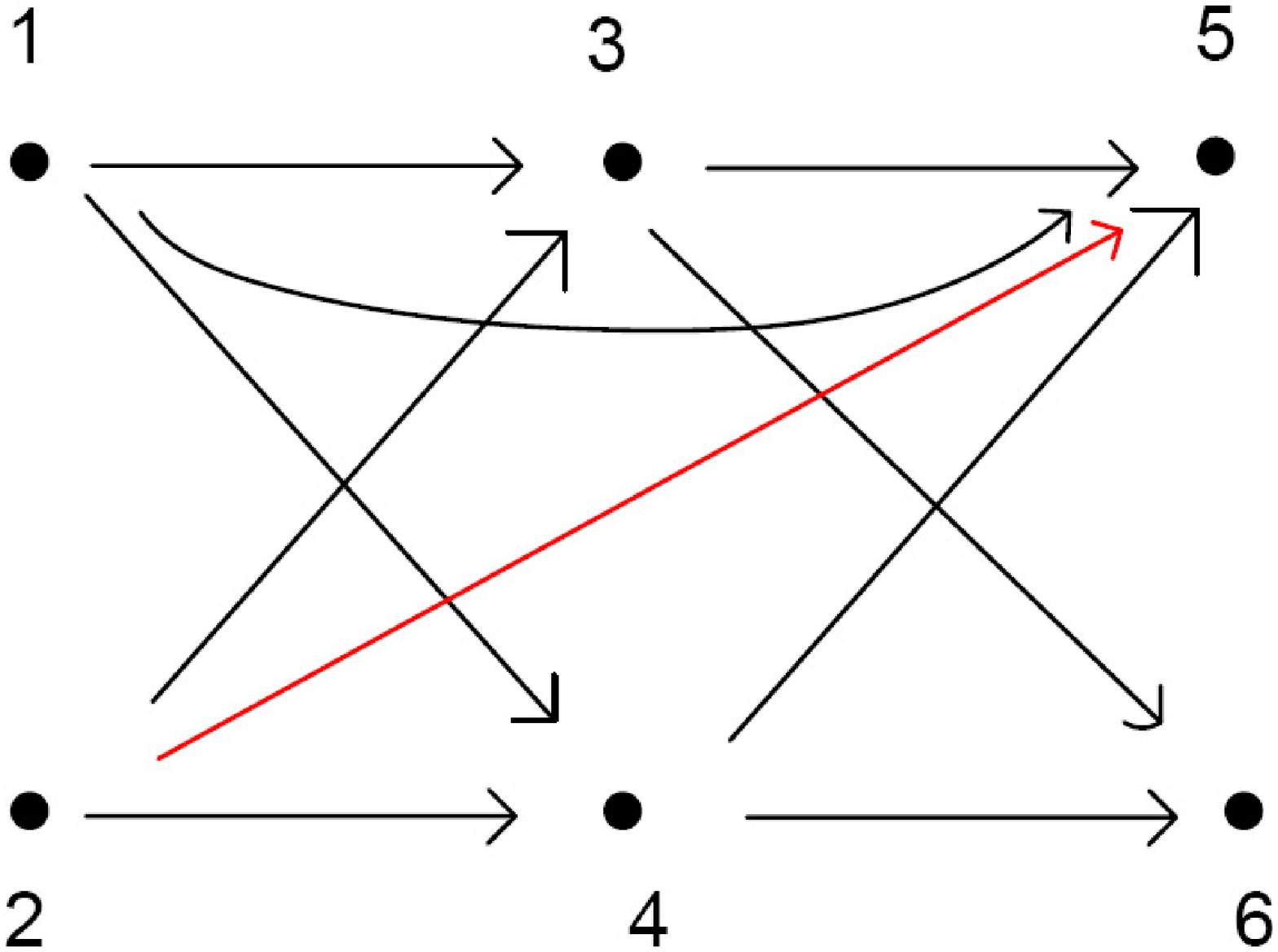}}
\subcaption{}
\label{diss}
\end{minipage}
\begin{minipage}[b]{.45\textwidth}
\hspace{.5cm}\scalebox{.3}{\includegraphics{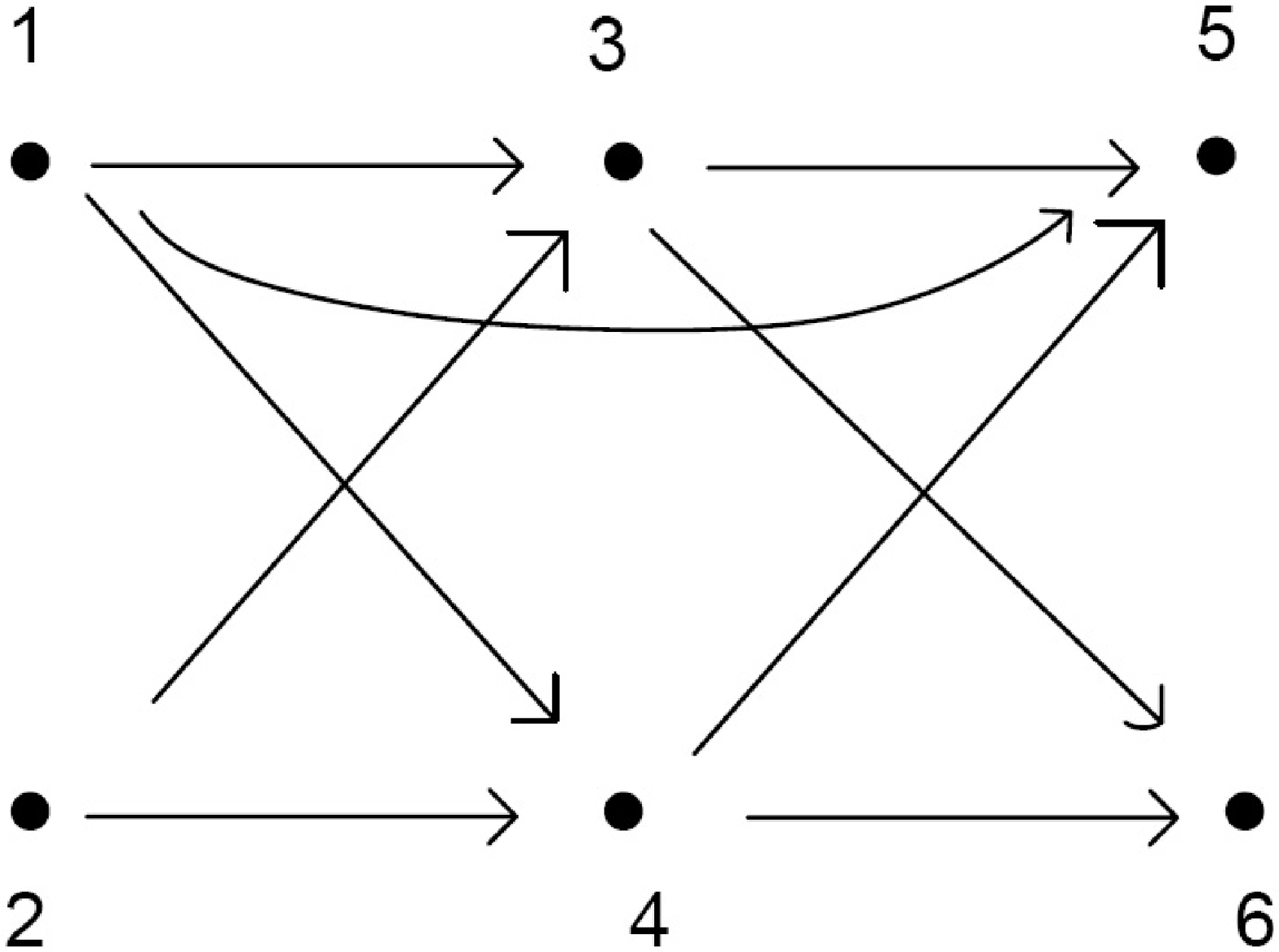}}
\subcaption{}
\label{diss2}
\end{minipage}
\caption{\label{tp} Two examples of networks: (a) satisfies the topological constraint, and (b) does not satisfy the topological constraint.}
\end{figure}

The optimal weight vector at node $5$ for both networks of Fig.\ref{tp} computed from (\ref{eq:wn}) is $w_5 = [-1, -1, 1, 1]$. This means that there should be a link between node $2$ and node $5$ for exact data incest removal according to the topological constraint (\ref{constraint}). Hence, Constraint~1 does not hold for the network of Fig.\ref{diss2}, while the topological constraint is satisfied in network depicted in Fig.\ref{diss}. Also as it is clear from the network shown in Fig.\ref{diss}, there is no need for the communication graph to be a tree.

\subsection{Discussion of Data Incest Removal in Social Learning}\label{subsec:dis} Here, we discuss the application of data incest removal Algorithm~1 (presented in Theorem~\ref{prop:mis}) in two examples of multi-agent state estimation problem which are presented in Sec.\ref{subsec:context}: (i) online reputation systems, and (ii) target localization using social networks, see Fig.\ref{BD4}. As explained in Sec.\ref{sec:intro}, data incest makes the estimates of the underlying state of nature biased in these two examples. Both problems can be formulated using the five-step constrained social learning protocol presented in Sec.\ref{subsec:pf}. As illustrated in Fig.\ref{model}, agents observe the underlying state of nature in noise  and practice social learning to choose an action such that a local cost function is minimized. But as a result of unknown communication graph  and the recursive nature of Bayesian estimators, data incest or abusive re-use of information occurs. To mitigate data incest, the network administrator plays an intermediating role. Instead of transmitting the communication graph and the set of all actions, the network administrator monitors all the information exchanges and provides the data incest free public belief of social learning at each node.
\begin{figure}[t]
\centering
\hspace{0cm}\scalebox{.3}{\includegraphics{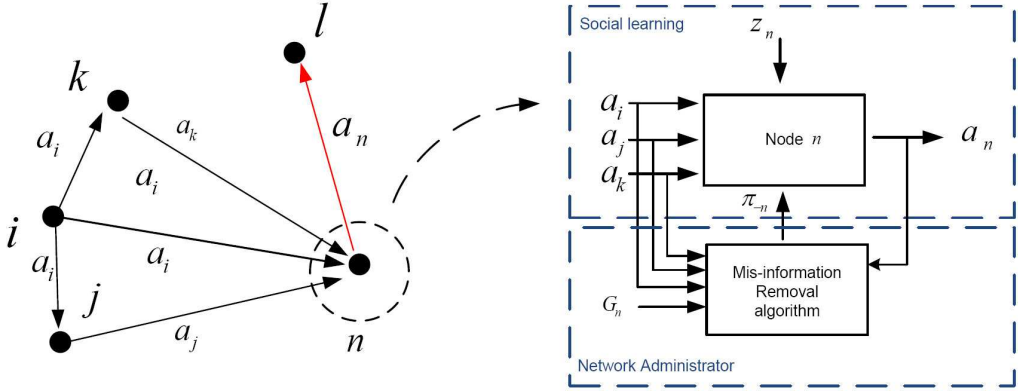}}
\caption{Data incest removal algorithm employed by network administrator in the state estimation problem over social network. The underlying state of nature could be geographical coordinates of an event (target localization problem) or reputation of a social unit (online reputation systems).}
\label{model}
\end{figure}
To compute the data incest free public belief, the network administrator uses the optimal information aggregation scheme (\ref{def:constraintestimate}) with  the optimal weight vector $w_n$, see (\ref{eq:wn})\footnote{Also as discussed in Sec.\ref{sec:intro}, an alternative method is to change action $a_n$ to another action $a^*_n$ such that the new after-action public belief is similar to that computed via  (\ref{def:constraintestimate}).}. Using the most updated public belief and its own private observation $z_n$, node $n$ evaluates its private belief. Based on this private belief (which is free of data incest), action $a_n$ is chosen and transmitted it over the network. Theorem~\ref{prop:mis} ensures that using the optimal weight vector $w_n$ and given that the communication graph satisfies the topological Constraint 1, the action $a_n$ is not affected by data incest and performance of the state estimation is improved.\\

\textbf{Remark 5:} The results of this paper can be applied to a scenario where the network administrator provides the after-action public beliefs to agents (instead of actions or the updated public belief). In this scenario, agents combine the received after-action public beliefs using the optimal weight vector $w_n$ to compute the updated public belief and then evaluate their private belief accordingly.

\section{Numerical Examples}\label{sec:num}
In this section, numerical examples are given to illustrate the performance of data incest removal Algorithm~1 presented in   Sec.\ref{sec:optimal}. As described in the five-step protocol of Sec.\ref{sec:intro}, agents interact on a graph to estimate an underlying state of nature (which represents the location of a target event in target localization problem, or the reputation of a social unit in online reputation systems). The underlying state of nature $x$ is a random variable uniformly chosen from $\mathbf{X} =  \{1,2,\cdots,20\}$, and actions are chosen from $\mathbf{A} = \{1,2,\ldots 10\}$. We consider the following three scenarios for each of four different types of social networks: \begin{enumerate}[(i)]
\item Constrained social learning without data incest removal algorithm (data incest occurs) depicted with dash-dot line
\item Constrained social learning with Protocol~1 with data incest removal algorithm depicted with dashed line
\item Idealized framework where each node has the entire history of raw observations and thus data incest cannot propagate. This scenario is only simulated for comparison purposes and is depicted by solid line.
\end{enumerate}
The effect of data incest on estimation problem and the performance of the data incest removal algorithm, proposed in Sec.\ref{sec:optimal}, is investigated for the networks shown in Fig.\ref{csomm}.

We first consider a communication graph with $41$ nodes. The communication graph under study, which is shown in Fig.\ref{CG}, satisfies the topological constraint (\ref{constraint}). The action of node 1 reaches all other nodes and node 41 receives all actions of previous 40 nodes (some edges are omitted from the figure to make it more clear).
\begin{figure}[htb]
\begin{minipage}[b]{0.33\textwidth}
\centering
\hspace{0cm}\scalebox{1}{\includegraphics[width=\textwidth]{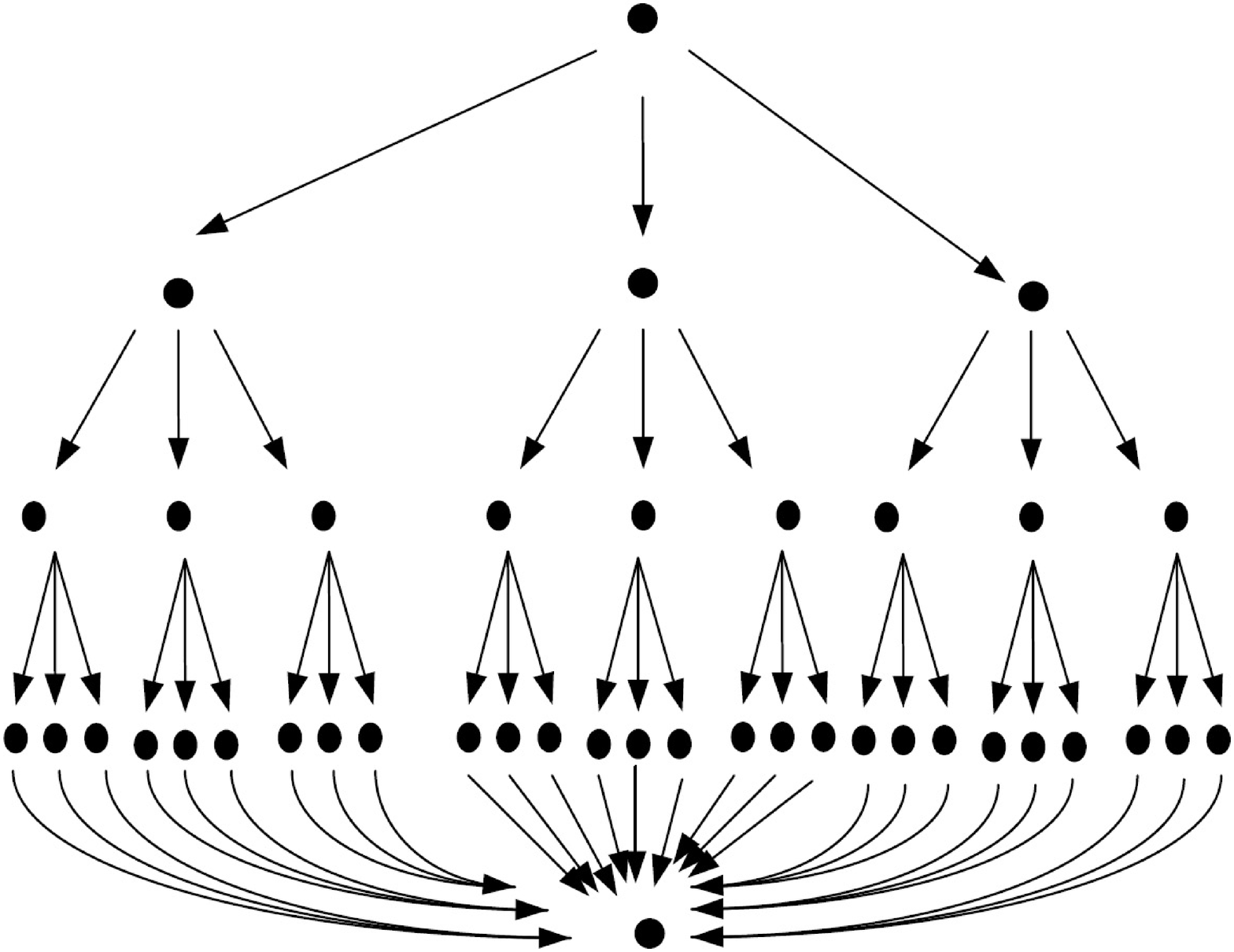}}
\subcaption{}
\label{CG}
\end{minipage}
\begin{minipage}[b]{0.3\textwidth}
\centering
\hspace{+1.2cm}\scalebox{.9}{\includegraphics[width=\textwidth]{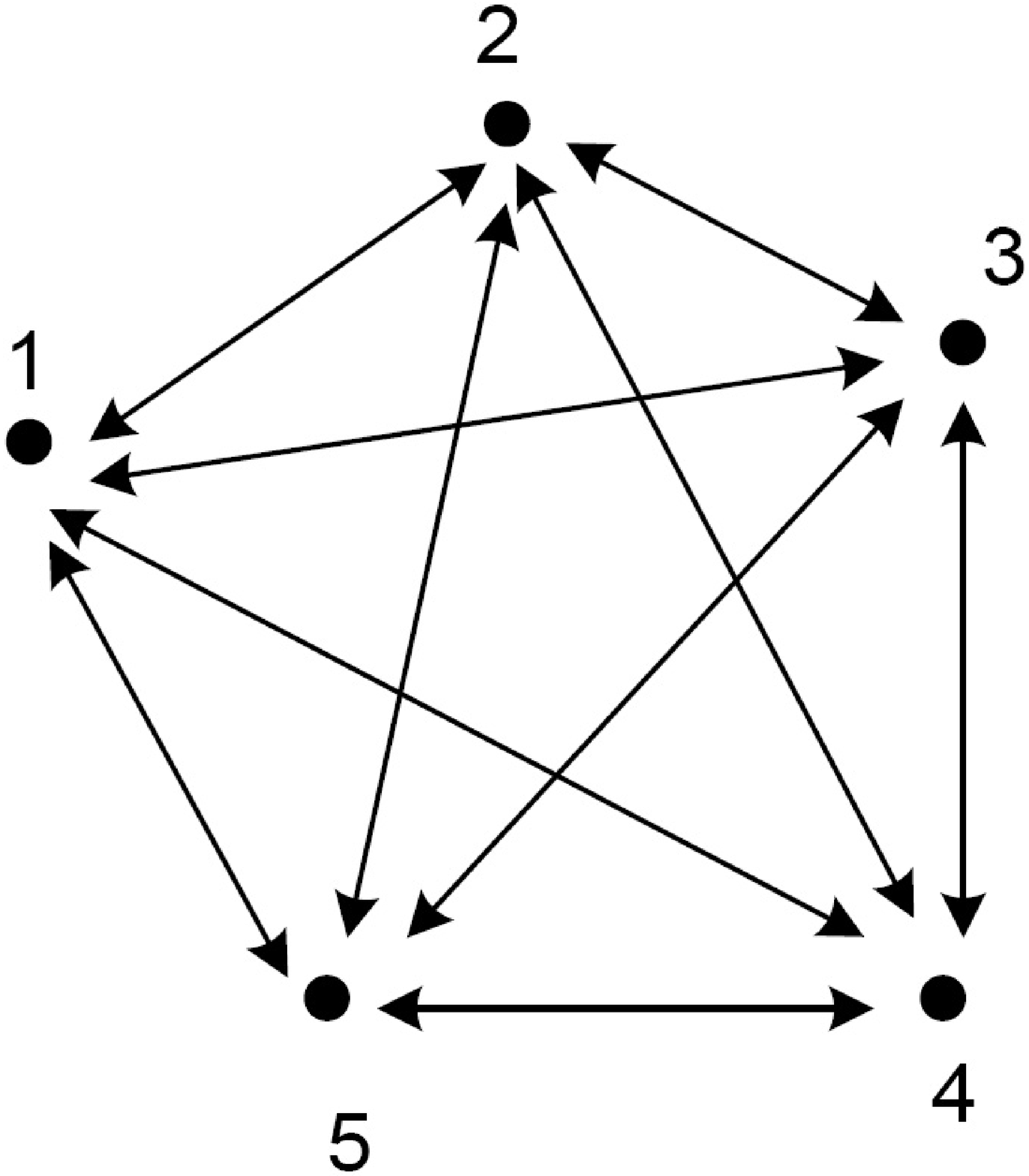}}
\subcaption{}
\label{com}
\end{minipage}
\begin{minipage}[b]{0.3\textwidth}
\centering
\hspace{+1.2cm}\scalebox{.9}{\includegraphics[width=\textwidth]{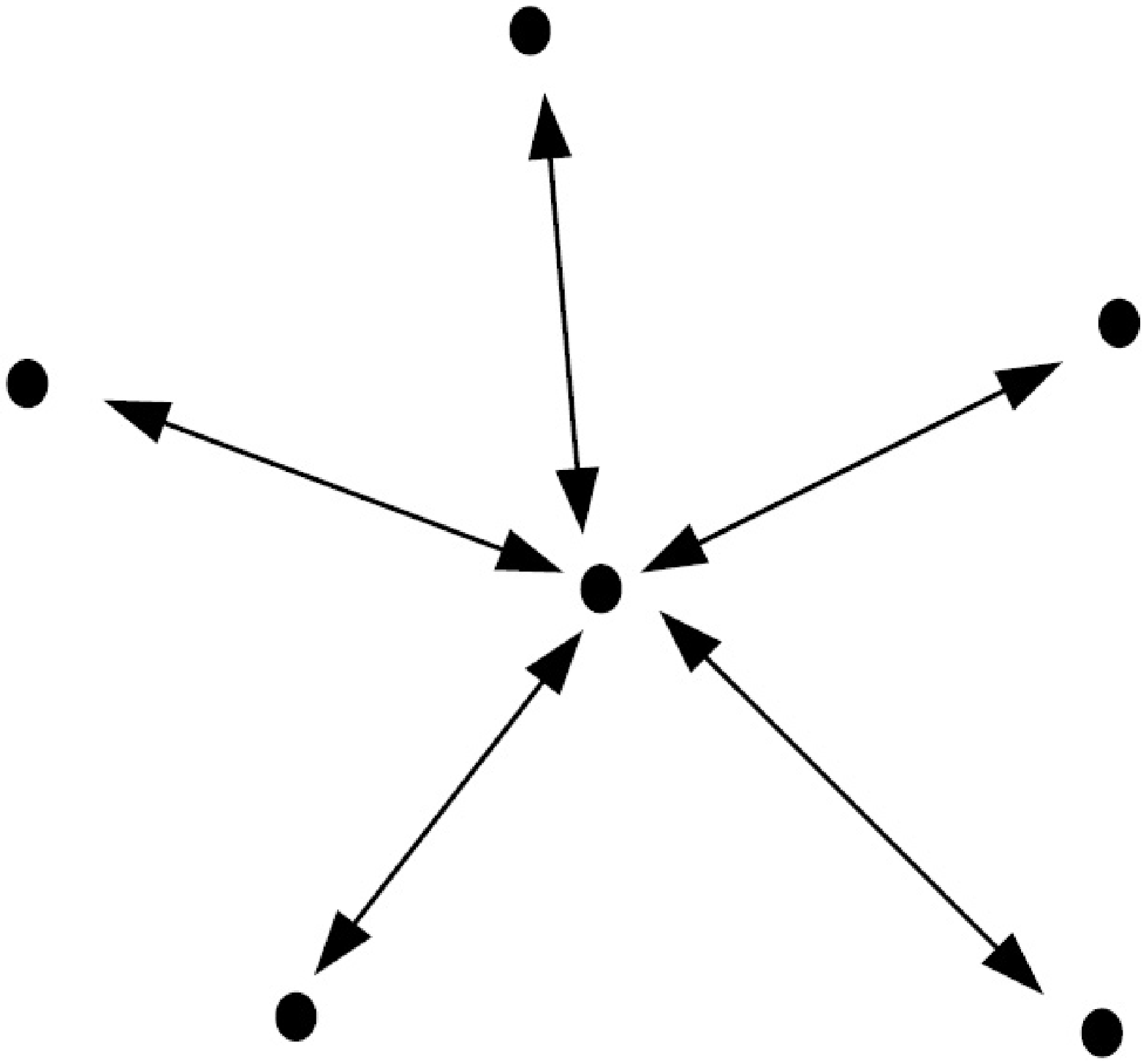}}
\subcaption{}
\label{cgg}
\end{minipage}

\caption{\label{csomm} Three different communication topologies: (a) the communication graph with $41$ nodes, (b) agents interact on a fully interconnected graph and the information from one agent reach other agents after a delay chosen randomly from $\{1,2\}$ with the same probabilities, (c) star-shaped communication topology with random delay chosen from $\{1,2\}$.}
\end{figure}

As can be seen in Fig.\ref{act-max}, data incest makes agents' actions in the constrained social learning without data incest removal different from the same in the  idealized framework. Also Fig.\ref{act-max} corroborates the excellent performance of data incest removal Algorithm~1.  As illustrated in Fig.\ref{act-max}, the actions of agents in social learning with data incest removal algorithm are exactly similar to those of the idealized framework without data incest.  The social learning problem over the graph shown in Fig.\ref{CG} is simulated 100 times to investigate the difference between  the estimated state of nature with the true one ($x=10$). The estimates  of state of nature (obtained in three different scenarios discussed in the beginning of the section) are depicted in Fig.\ref{mean-max}. As can be seen from the figure, the estimates obtained with data incest removal algorithm are very close to data incest free estimates of Scenario (iii). The bias in estimates in presence of data incest is also clear in this figure.
\begin{figure}[htb]
\begin{subfigure}[h]{0.5\textwidth}
\centering
\hspace{-1.5cm}\scalebox{.6}{\includegraphics{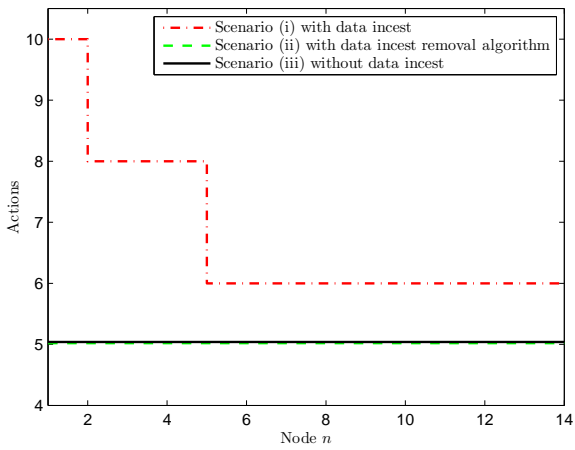}}
\caption{}
\label{act-max}
\end{subfigure}
\begin{subfigure}[h]{0.5\textwidth}
\centering
\hspace{+3cm}\scalebox{.6}{\includegraphics{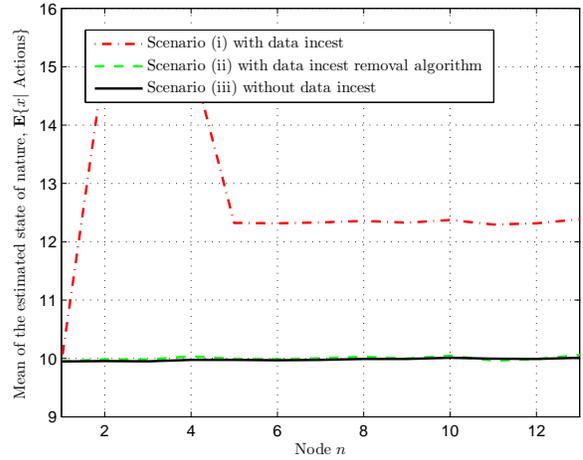}}
\caption{}
\label{mean-max}
\end{subfigure}
\caption{(a) Actions of agents obtained with social learning over social networks in three different scenarios (constrained social learning without data incest removal, constrained social learning with data incest removal algorithm, and idealized data incest free social learning) with communication graph depicted in Fig.\ref{CG}, (b) mean of the estimated state of nature in the state estimation problem with the same communication graph.}
\end{figure}

In the next simulation, a different communication topology is considered. We repeat the simulation for a star-shaped communication graph comprising  of six agents ($S=6$) at four time instants, $K=4$, so the total number of nodes in the communication graph is $24$, see Fig.\ref{cgg}. The communication delay is randomly chosen from $\{1,2\}$ with the same probabilities. We simulated the social learning in three different scenarios discussed above, to investigate the effect of data incest on the actions and the estimates of agents in the star-shaped social network. The actions chosen by nodes are depicted in  Fig.\ref{act-star}.  As can be seen from Fig.\ref{act-star}, using the data incest removal algorithm,  the agents' actions in the constrained social learning with Protocol~1 are very close to those of the idealized social learning with Protocol~2 which are free of data incest. Also the estimates of state of nature are very close to the true value of state of nature compared to the constrained social learning without data incest removal algorithm. Also note that the effect of data incest, as expected, in this communication topology is different for each agent; the agent who communicates with all other nodes is affected more by data incest. This fact is verified in Fig.\ref{starr}.
\begin{figure}[htb]
\begin{subfigure}[htb]{0.5\textwidth}
\centering
\hspace{-1.5cm}\scalebox{.6}{\includegraphics{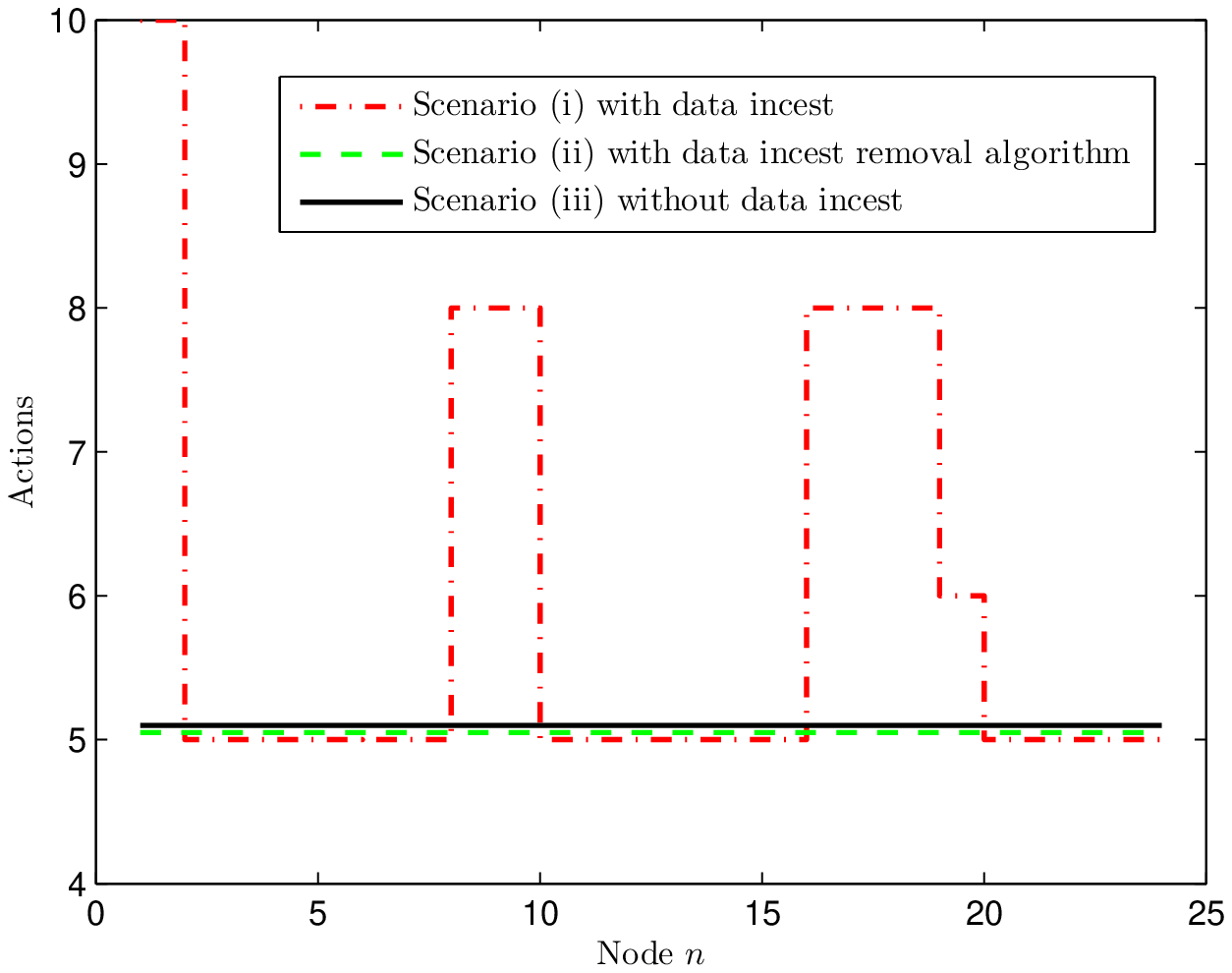}}
\caption{}
\label{act-star}
\end{subfigure}
\begin{subfigure}[htb]{0.5\textwidth}
\centering
\hspace{+3cm}\scalebox{.6}{\includegraphics{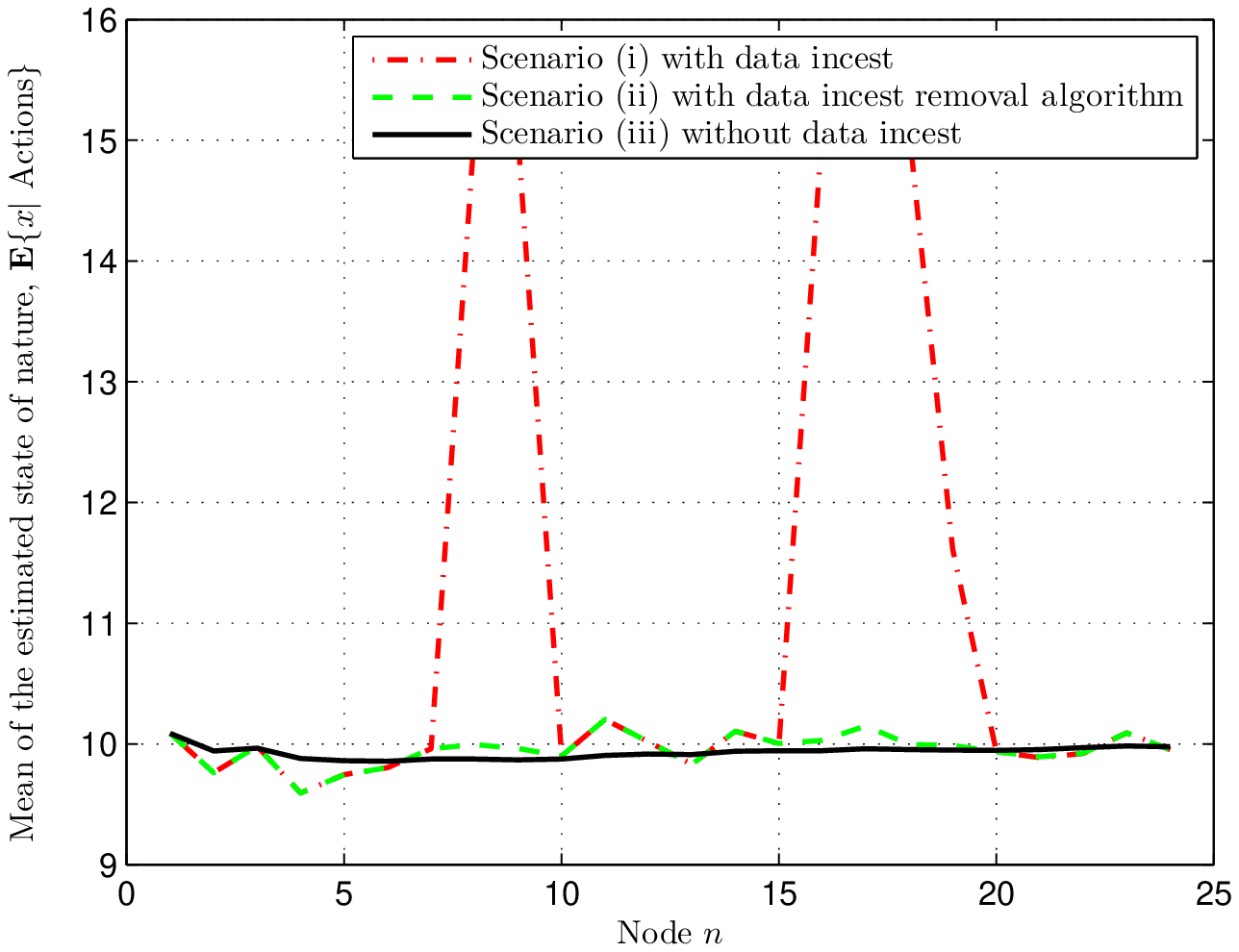}}
\caption{}
\label{mean-star}
\end{subfigure}
\caption{\label{starr}(a) Actions of agents obtained with social learning over social networks in three different scenarios (constrained social learning without data incest removal, constrained social learning with data incest removal algorithm, and idealized data incest free social learning) with communication graph depicted in Fig.\ref{cgg},(b) mean of the estimated state of nature in the state estimation problem with the same communication graph.}
\end{figure}

In the third example, a complete fully interconnected graph (where agents communicate with all other agents) is considered. In this example, action of each agent becomes available at all other agents after a random delay chosen from $\{1,2\}$ with the same probabilities. The agents' actions are shown in Fig.\ref{act-com}. Similar to the star-shaped graph, using data incest removal Algorithm~1 makes the agents' actions in the constrained social learning very similar to those of the idealized (data incest free) framework. Also, the excellent performance of  data incest removal Algorithm~1 in the estimation problem is depicted in Fig.\ref{mean-com}.
\begin{figure}[htb]
\begin{subfigure}[htb]{0.5\textwidth}
\centering
\hspace{-1.5cm}\scalebox{.6}{\includegraphics{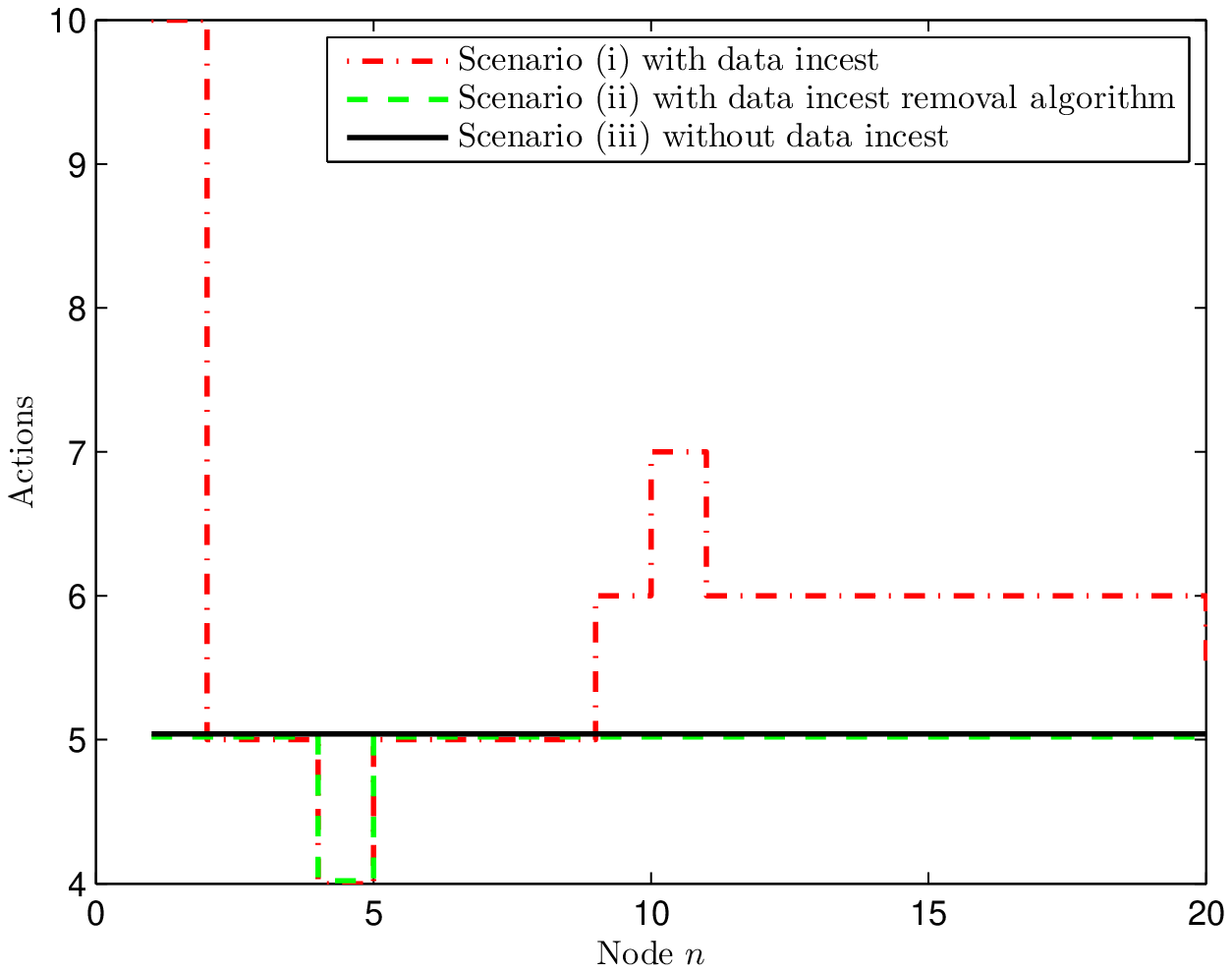}}
\caption{}
\label{act-com}
\end{subfigure}
\begin{subfigure}[htb]{0.5\textwidth}
\centering
\hspace{+3cm}\scalebox{.6}{\includegraphics{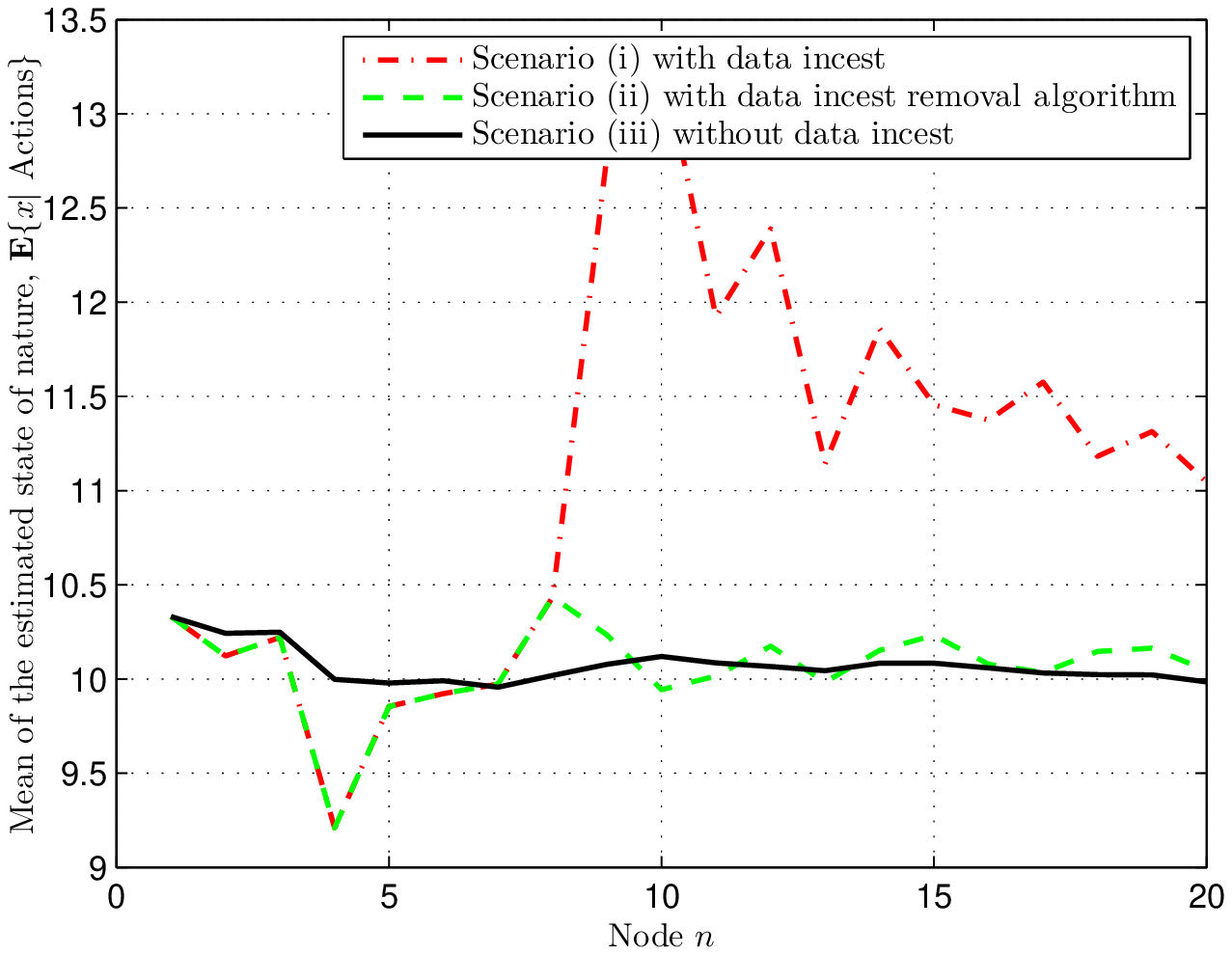}}
\caption{}
\label{mean-com}
\end{subfigure}
\caption{(a) Actions of agents obtained with social learning over social networks in three different scenarios (constrained social learning without data incest removal, constrained social learning with data incest removal algorithm, and idealized data incest free social learning) with communication graph depicted in Fig.\ref{com}, (b) mean of the estimated state of nature in the state estimation problem with the same communication graph.}
\end{figure}

We also extend our numerical studies to an arbitrary random network with five agents, $S=5, K=4$. We consider a fully connected  network  and assume that the interaction between two arbitrary agents (say agent $i$ and agent $j$) at time $k$ has four (equiprobable) possible statuses: (i) connected with delay 1, (ii) connected with delay 2, (iii) connected with delay 3, and (iv) not connected. If the link is connected with delay $\tau$, this means that the information from agent (i) at time $k$ becomes available at agent $j$ at time $k+\tau$. If the link is not connected, the information of agent $i$ at time $k$ never reaches agent $j$. We verify that the underlying communication graph, $G_n$, satisfies the topological Constraint~1 with simulation. Fig.\ref{act-rnd} depicts the agents' actions  in three different scenarios (with data incest, without data incest, and with data incest removal algorithm). The simulation results show that, even in this case with arbitrary network (that satisfies topological constraint), the actions obtained by the constrained social learning with data incest removal algorithm is very close to those in the idealized social learning.  As expected, using the data incest removal algorithm, the data incest associated with the estimates of agents can be mitigated completely, as shown in Fig.\ref{mean-rnd}.
\begin{figure}[!h]
\begin{subfigure}[b]{0.5\textwidth}
\centering
\hspace{-1.5cm}\scalebox{.6}{\includegraphics{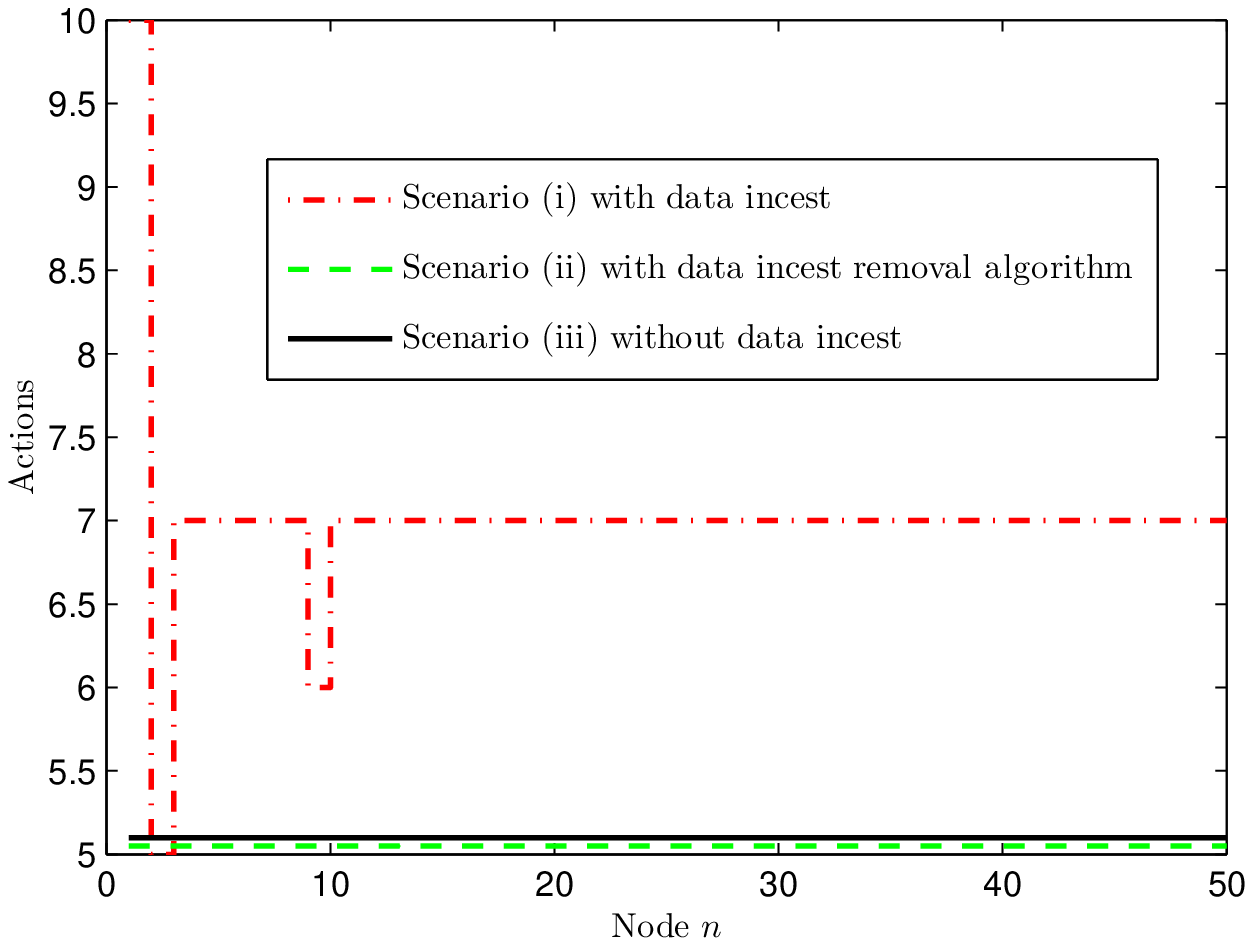}}
\caption{}
\label{act-rnd}
\end{subfigure}
\begin{subfigure}[b]{0.5\textwidth}
\centering
\hspace{+3cm}\scalebox{.6}{\includegraphics{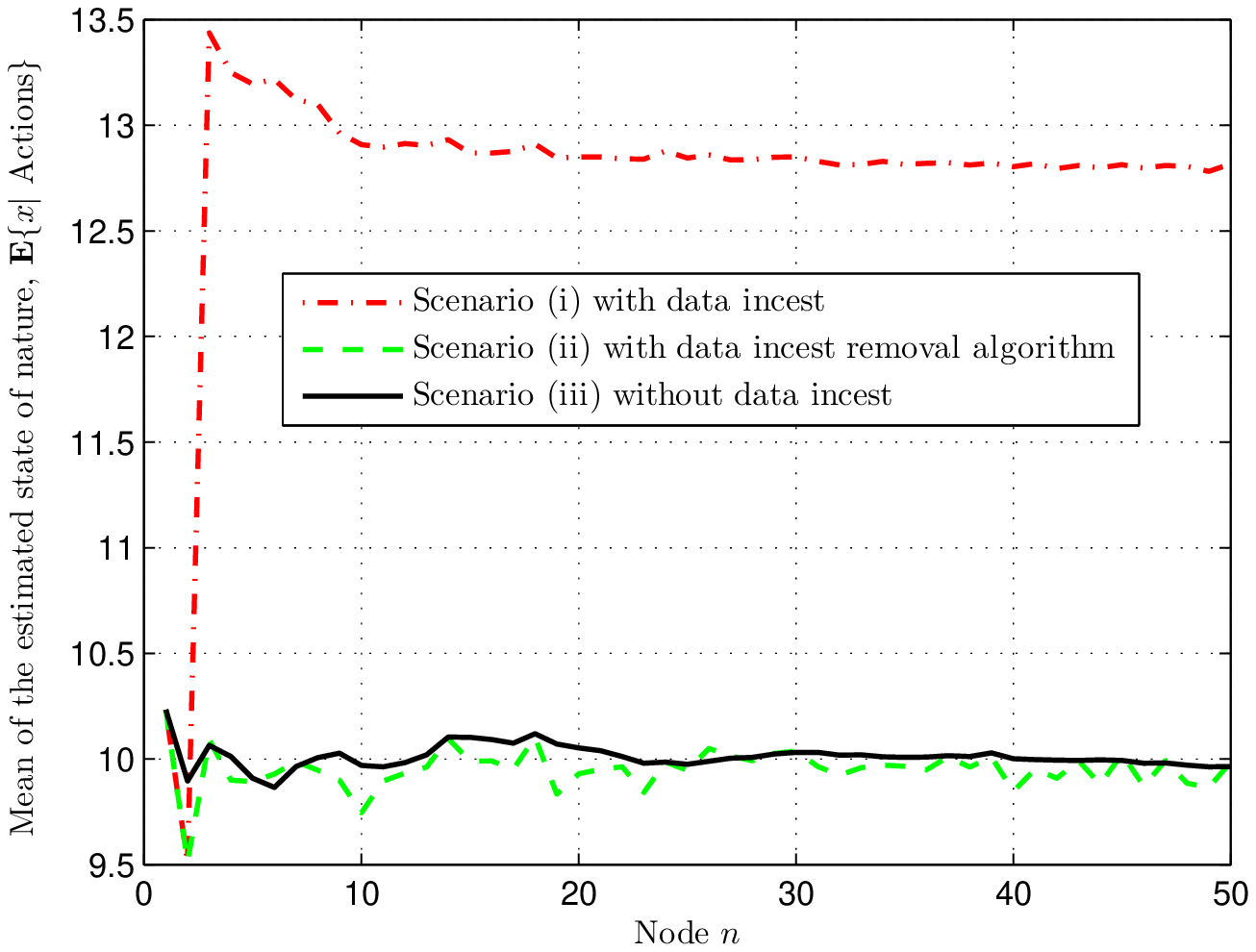}}
\caption{}
\label{mean-rnd}
\end{subfigure}
\caption{(a) Actions of agents obtained with social learning over social networks in three different scenarios (constrained social learning without data incest removal, constrained social learning with data incest removal algorithm, and idealized data incest free social learning) with arbitraty communication graph, (b) mean of the estimated state of nature in the state estimation problem with the same communication graph.}
\end{figure}

Here, we discuss the accuracy of the state estimates in terms of mean squared error with numerical studies. The mean squared error of estimates obtained in social learning with three different scenarios discussed in the beginning of this section (with data incest, with data incest removal algorithm, and the idealized framework) is computed for each of four communication graphs considered in our numerical studies. The results are depicted in Fig.\ref{mse}. As can be seen from Fig.\ref{mse}, the estimates of the constrained social learning with data incest removal Algorithm~1 are more accurate than the those of the constrained social learning in presence of data incest. However, as a result of herding, in star shaped and random communication topologies, the mean squared error of estimates is slightly (compared to the scenario without data incest removal algorithm) more than the idealized framework at each time.
\begin{figure}[!h]
\begin{minipage}{\textwidth}
\begin{minipage}[b]{0.5\textwidth}
\centering
\hspace{-1.5cm}\scalebox{.6}{\includegraphics{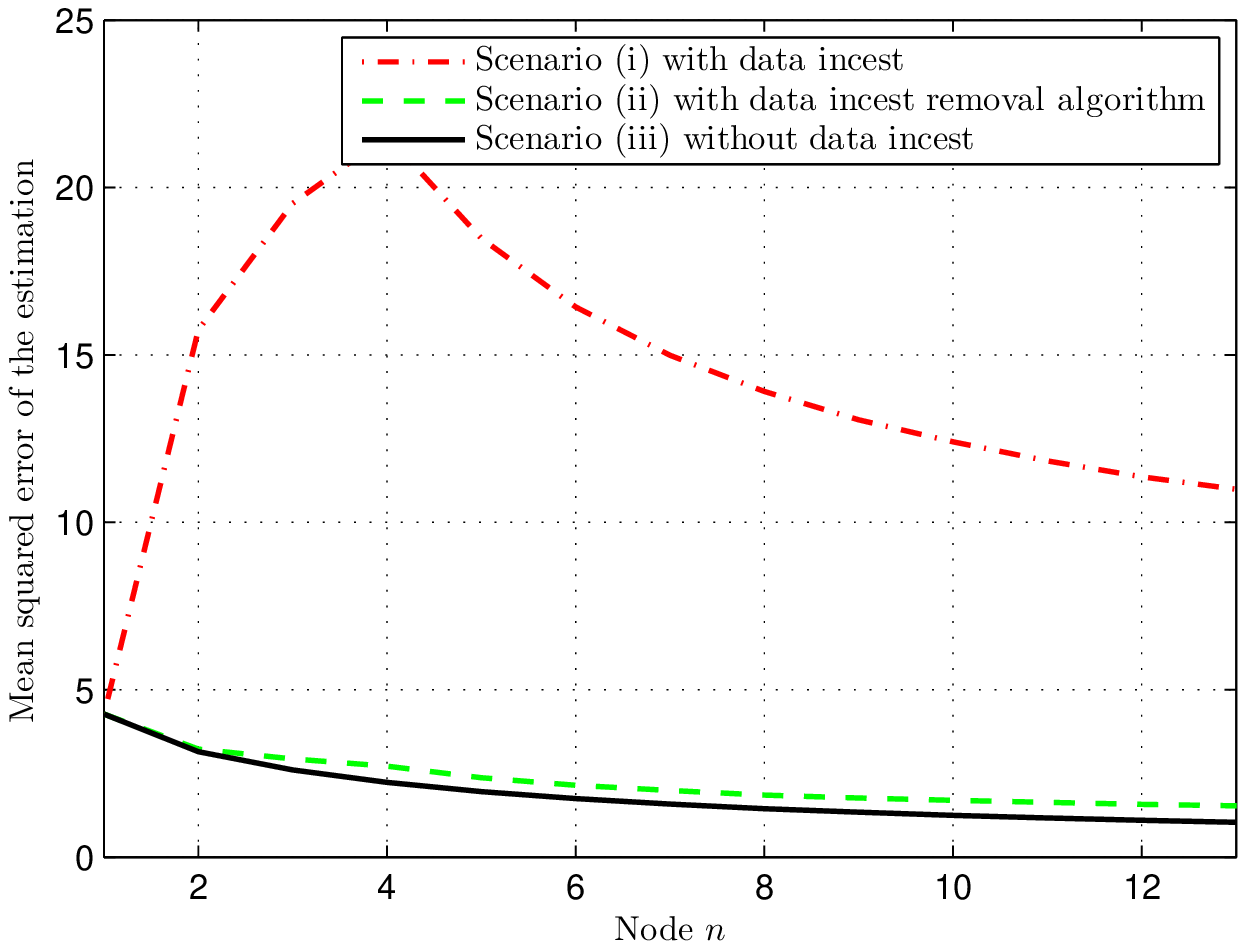}}
\subcaption{}
\label{mse-max}
\end{minipage}
\begin{minipage}[b]{0.5\textwidth}
\centering
\hspace{+3cm}\scalebox{.6}{\includegraphics{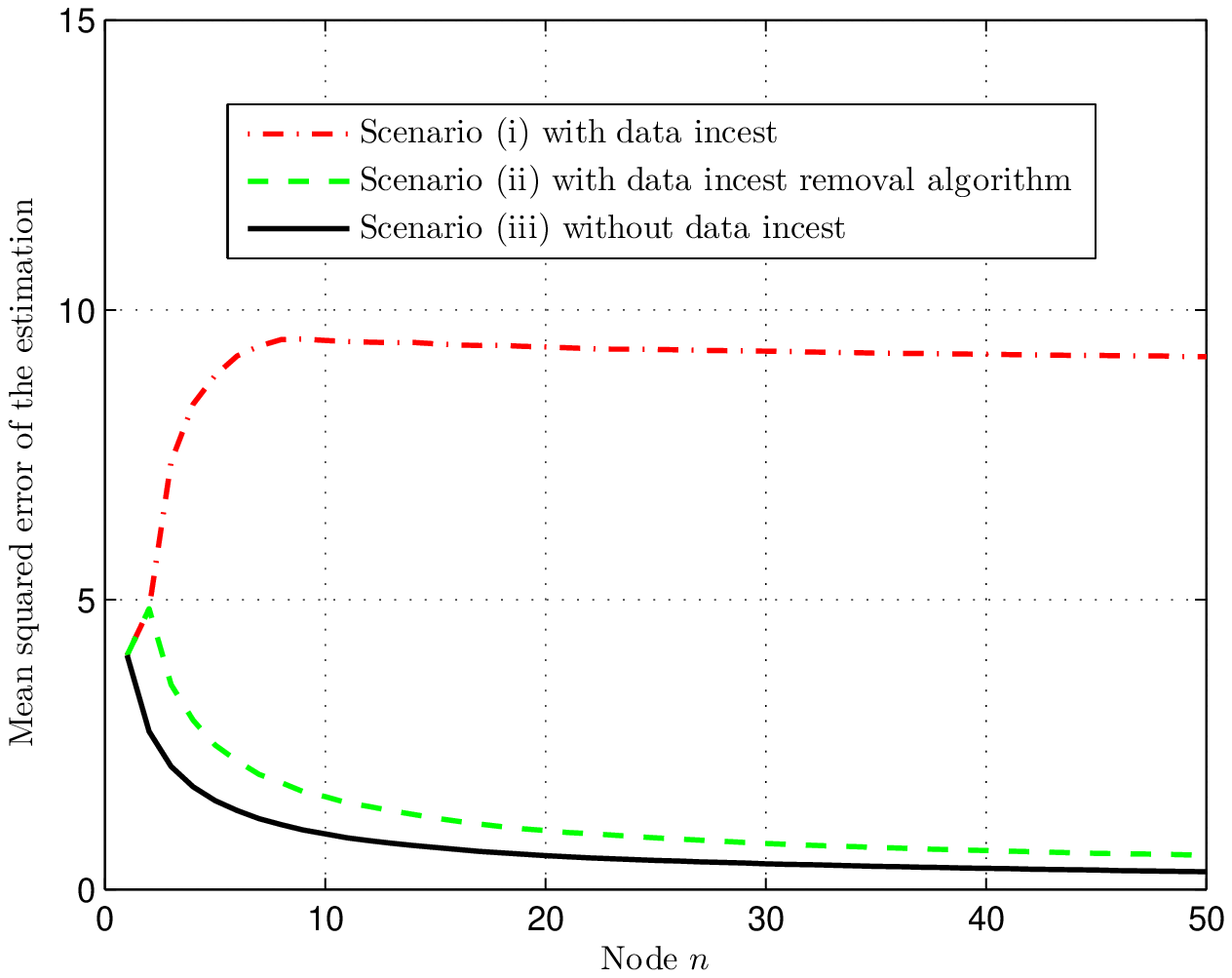}}
\subcaption{}
\label{mse-rnd}
\end{minipage}
\end{minipage}
\vspace{2mm}
\begin{minipage}{\textwidth}
\begin{minipage}[b]{0.5\textwidth}
\centering
\hspace{-1.5cm}\scalebox{.6}{\includegraphics{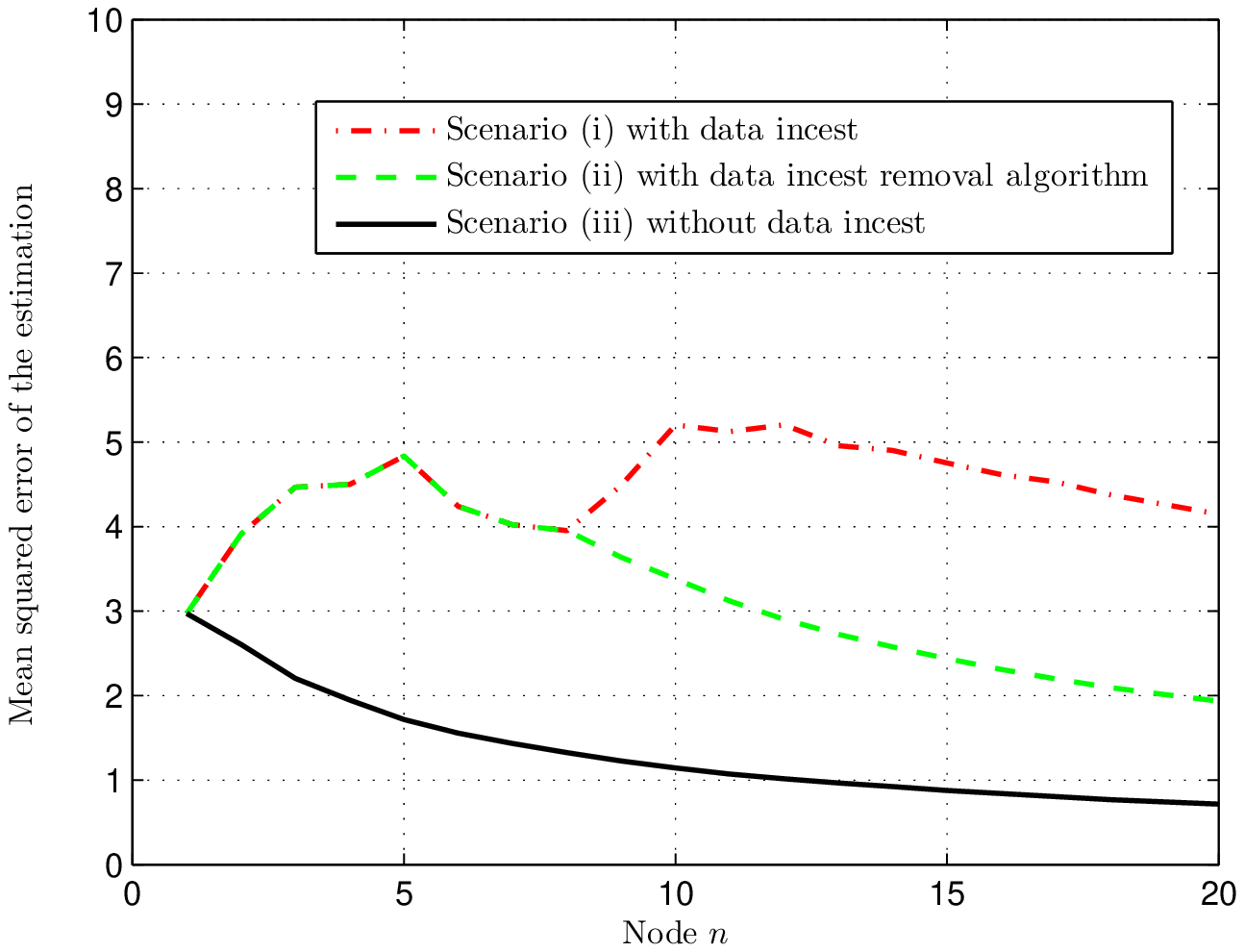}}
\subcaption{}
\label{mse-com}
\end{minipage}
\begin{minipage}[b]{0.5\textwidth}
\centering
\hspace{+3cm}\scalebox{.6}{\includegraphics{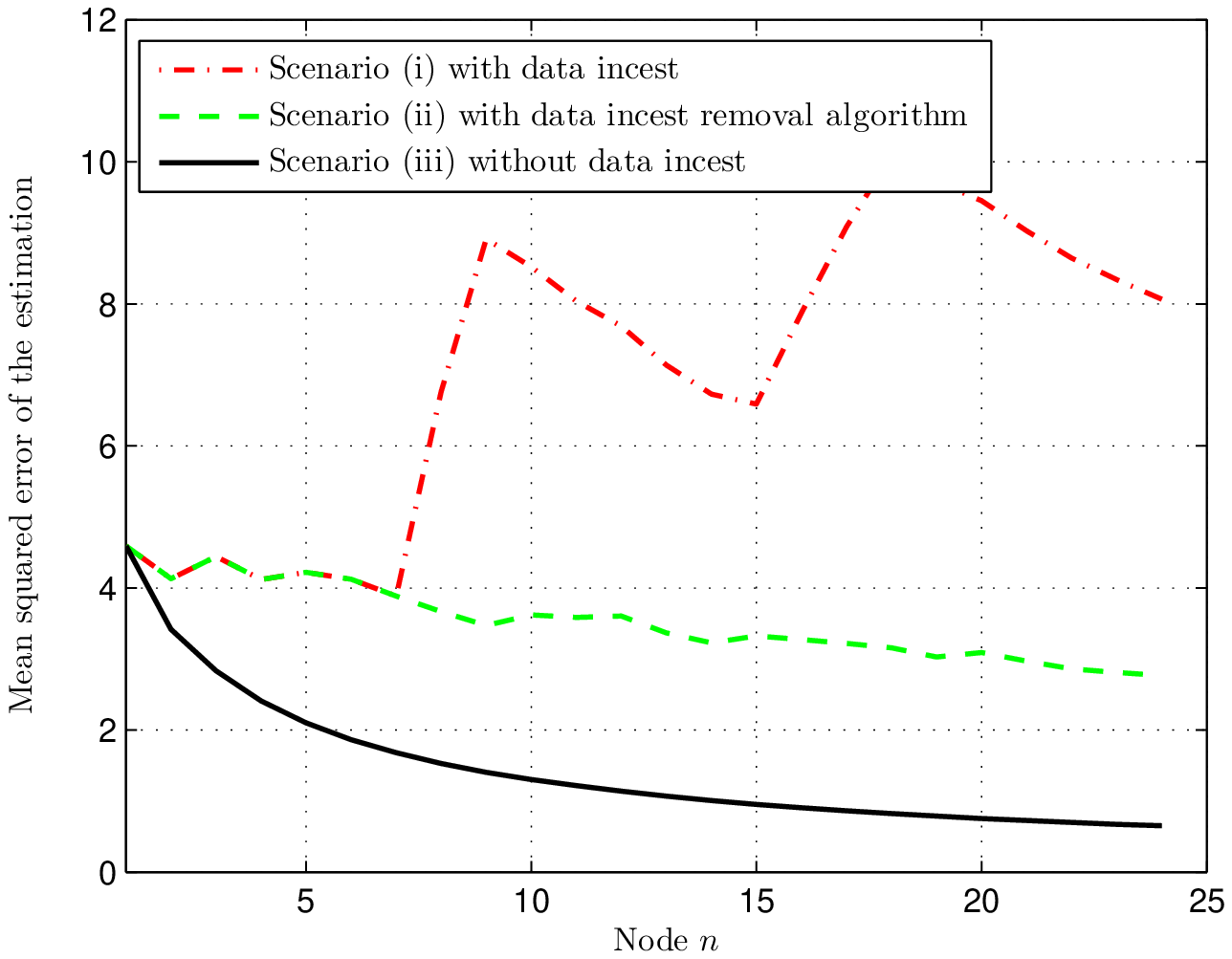}}
\subcaption{}
\label{mse-star}
\end{minipage}
\end{minipage}
\caption{\label{mse} Mean squared error of estimates (of state of nature) obtained with social learning with (a) communication graph depicted in Fig.\ref{CG}, (b) complete fully interconnected graph depicted in Fig.\ref{com}, (c) arbitrary communication graph, and (d) star-shaped communication graph depicted in Fig.\ref{cgg}.}
\end{figure}
\section{Conclusion}\label{sec:con}
In this paper, the state estimation problem in social networks with social learning is investigated. The state of nature could be geographical coordinates of an event (target localization problem) or quality of a social unit (online reputation system).  As discussed in the paper, data incest arises in this setup  as a result of the recursive nature of Bayesian estimation and random communication delays in social networks. We proposed a data incest removal algorithm for the multi-agent social learning in social networks in this paper along with a topological necessary and sufficient condition for data incest free estimation. The main difference of this work with the existing data incest removal algorithms in the literature is that in this paper we consider the data incest removal algorithm in social learning context where only public belief of agents (which can be computed directly from actions) is transmitted over the network while in existing data incest removal algorithms (in sensor networks or social networks) the private belief of agents which depends on their private observations are transmitted through the network. In the future work, we will consider the application of the results of this paper in the problem of inference in the graphical models using message passing algorithms.
\appendix
\subsection{Some Graph Theoretic Definitions}\label{subsec:graphtheory}

\textbf{Graph, Directed Graph, Path and Directed Acyclic Graph (DAG):}\begin{itemize}
\item A graph $G_N$ comprising of $N$ nodes is a pair $(V,E)$, where $V=\{v_1,\ldots,v_N\}$ is a set of nodes (also called vertices), and $E \subset V \times V $ is a set of edges between the nodes. \item Graph $G_N$ is an undirected graph if for any $(v_i,v_j)\in E$ then $(v_j,v_i)\in E $ and a graph is said to be directed if $(v_j,v_i)\in E $ is not a consequence of $(v_i,v_j)\in E$. \item A \textit{path} is an alternating sequence of nodes and edges, beginning and ending with an edge, in which each node is incident to the two edges that precede and follow it in the sequence.\item A \textit{Directed Acyclic  Graph (DAG)} is a directed graph with no path that starts and ends at the same node.
    \item A family of DAGs $\mathcal{G}_N$ is defined as a set of DAGs $\{G_1,\ldots,G_N\}$ where $G_{n}$ is the sub-graph of
$G_{n+1}$ such that for $ n = 1,\ldots,N-1$
\begin{equation}\label{eq:familyofDAGs}
\left\{\begin{array}{l}
V_n= V_{n+1}/v_{n+1}\;,\\
E_n = E_{n+1}/\{(v_i,v_{n+1})\in E_{n+1}| v_i \in V_{n+1}\}\;.
\end{array}\right.
\end{equation}\end{itemize}

\textbf{Adjacency and Transitive Closure matrices:}

Let $G_N =(V ,E )$ denote a graph with $N$ nodes $V =\{v_1,\ldots,v_N\}$.\begin{itemize}
\item The Adjacency Matrix $\bar A $ of $G_N $ is an $N\times N$ matrix whose elements $\bar A(i,j)$ are given by \begin{eqnarray}\label{eq:adjacencymatrix}
\bar A (i,j)=\left\{
\begin{array}{ll}
1 &\textrm{ if } (v_j,v_i)\in E \;, \\
0 &\textrm{ otherwise}
\end{array} \right.\;. \text{  } A(i,i)=0.
\end{eqnarray}
\item The Transitive Closure Matrix ${{T}}$ of $G_N$ is an $N\times N$ matrix whose elements ${{T}}(i,j)$ are given by  $T(i,i)=1$, and
\begin{eqnarray}\label{eq:transitivieclosurematrix}
{{T}} (i,j)=\left\{
\begin{array}{ll}
1 &\textrm{ if there is a path between } v_j \textrm{ and } v_i\;, \\
0 &\textrm{ otherwise}
\end{array} \right.\;.
\end{eqnarray}
\end{itemize}

The following shows the special form of the adjacency matrix of the directed acyclic graph and provides a closed form expression to compute the transitive closure matrix from the adjacency matrix of a directed acyclic graph.
\begin{Lemma} \label{lem:am2tcm}
\it
A sufficient condition for a graph $G_N$ to be a DAG is that  the Adjacency matrix $A$ is an upper triangular matrix. For a DAG $G_N$, the Transitive Closure Matrix  ${{T}}$ is related to the Adjacency matrix by
\begin{eqnarray}\label{eq:am2tcm}
{{T}} = Q(\{{\mathbf{I}}_N-A\}^{-1}).
\end{eqnarray}
\end{Lemma}
Here, ${\mathbf{I}}_N$ is the $N\times N$ identity matrix, and  $Q$ denote the matrix valued "quantization"  function so that for any $N\times N$-matrix $B$, $Q(B)$ is the $N\times N$ matrix with elements
\begin{eqnarray}\label{eq:quantizationfunction}
Q(B)(i,j) = \begin{cases} 0 & \text{ if } B(i,j)=0\;, \\
1  & \text{ if } B(i,j) \neq 0 \end{cases}.
\end{eqnarray}

{\it Proof:} This result is derived from the classical interpretation of matrix $\{{\mathbf{I}}_N-A\}^{-1}$. The entry in row $i$ and column $j$ of this matrix gives the number of paths from node $i$ to node $j$. \hfill $\square$\\

 To deal with information flow in a social network, we now introduce the concept
of a {\em family} of DAGs.

\textbf{Remark 6:}
For the sake of simplicity in notations, let us define two vector representatives of adjacency and transitive closure matrices of directed acyclic graph. For each graph $G_n \in \mathcal{G}_N$, let the $n\times n$ matrices $\bar A_n$ and $T_n$, respectively, denote the adjacency matrix and transitive closure matrix. Define the following:
\begin{eqnarray}\label{deft}
\left\{\begin{array}{l}
\text{ $t_n \in \{0,1\}^{1 \times (n-1)}$: transpose of
first $n-1$  elements of $n$th column of $T_n$,}
\label{eq:tn}
\\
\text{ $b_n \in \{0,1\}^{1 \times (n-1)}$: transpose of first $n-1$ elements
of $n$th column of $\bar A_n$}.
\end{array}\right.
\end{eqnarray}

\textbf{Remark 7:}
As can be straightforwardly followed from the construction of adjacency and transitive closure matrices in (\ref{eq:familyofDAGs}), for a family of DAGs $\mathcal{G}_N=\{G_1,\ldots,G_N\}$, for any $n\in \{1,\ldots,N-1\}$, the adjacency matrix  $\bar A_n$  and transitive closure matrix ${T}_n$ of graph $G_n$ are respectively the $n\times n$ left upper matrices of the adjacency matrix $A_{n+1}$ and transitive closure matrix ${T}_{n+1}$ of graph $G_{n+1}$.
\subsection{Proof of Theorem \ref{Theorem:informationflowDAG} }\label{subsec:proofp1}
To prove that the graph $G_n = (V_n, E_n)$ from family $\mathcal{G}_n$ is a directed acyclic graph, we only need to show that the adjacency matrix of $G_n$ is an upper triangular matrix. Then from Lemma~\ref{lem:am2tcm}, the graph $G_n$ is a directed acyclic graph. Suppose that $v_i$ and $v_j$ are two vertices of $G_n$, that is $v_i, v_j \in V_n$. From re-indexing scheme (\ref{reindexing_scheme}), $v_i$ and $v_j$ represents agents $s_i$ and $s_j$ at time instants $k_i$ and $k_j$, respectively. We have $v_i = s_i +S(k_i-1)$ and $v_j = s_j + S(k_j - 1)$. Because of the information flow, information of each agent may become available at other agents at later time instants, a message cannot travel back in the time! This means that if $k_i < k_j$, there should not be an edge from $v_j$ to $v_i$, $(v_j,v_i) \notin E_n$. Using re-indexing scheme if $k_i < k_j$, then $v_i < v_j$ (because $k_i$ and $k_j$ are integers and $s_i, s_j \leq S$).
Therefore, we deduce that \begin{equation} 
  i<j   \Rightarrow (v_j,v_i)\notin E_n.
\end{equation}
Consequently, the Adjacency Matrix is a strictly upper triangular matrix so that $G_n$ is a DAG. Then it follows from the
construction of the DAGs that ${\mathcal G}_N$ is a family of DAGs.
\subsection{Proof of Lemma~\ref{lem:1}}\label{ap:lem1}
\begin{proof}
We assume that each node has the most up-to-date public belief of social learning, $\pi_{-n} = p(x|\Theta_n)$. This node records its own private observation $z_n = \bar z_l$. The private belief is
\begin{align}\label{eq:ap6}
\mu_n(m) =  p(x = \bar x_m|\Theta_n,z_n).
\end{align}
Using Bayes' theorem, (\ref{eq:ap6}) can be written as
\begin{align}\label{eq:ap6}
\mu_n(m) =   p(x = \bar x_m|\Theta_n,z_n) = cp(z_n|x)p(x|\Theta_n)\nonumber\\
=c \pi_{-n}(m)B_{ml}.
\end{align}
The normalizing factor $c$ is used to make $\mu_n$ a true probability mass function, that is $\sum_{m=1}^X \mu_n(m) = 1$. Expected cost given $\mu_n$ is equal to $C'_a\mu_n$ thus the action $a_n$ is $a_n = \argmin_{a\in \mathbb{A}}\{C_a'\mu_n\}$. To complete the proof, we need to compute the after-action public belief, $\pi_{+n} = p(x|\Theta_n, a_n)$. Applying Bayes' theorem, the after action public belief can be written as
\begin{align}\label{eq:ap9}
\pi_{+n}(m) &= p(x = \bar x_m|\Theta_n, a_n) = c p(a_n|\Theta_n,x) p(x = \bar x_m|\Theta_n) \nonumber\\&= c p(a_n|x,\pi_{-n})\pi_{-n}(m) = c\sum_{j =1}^{Z}{p(a_n|x, z = \bar z_j,\pi_{-n})}p( z = \bar z_j).
\end{align} Knowing observations and public belief, the private belief can be computed. From the private belief, the action $a_n$ is speified. Thus
\begin{align}\label{eq:ap7}
p(a_n|x, z = \bar z_j,\pi_{-n}) = \left\{\begin{array}{l}1 \quad \text{if } a_n = \argmin_{a\in \mathbb{A}}\{C_a'B_j\pi_{-n}\} \\ 0\quad \text{if } a_n \neq \argmin_{a\in \mathbb{A}}\{C_a'B_j\pi_{-n}\}  \end{array}\right.
\end{align}
where $B_j = {\rm diag}(B_1j,\ldots, B_{Xj})$. Using indicator function $\mathbb{I}(\cdot)$, Eq. (\ref{eq:ap7}) can be reorganized as
\begin{align}\label{eq:ap8}
p(a_n|x, z = \bar z_j,\pi_{-n}) = \prod_{\hat a \in \mathbf{A}-\{a_n\}}\mathbb{I}(C_{a_n}'B_{j}\pi_{-n} <C_{\hat a}'B_{j}\pi_{-n} )
\end{align}
Substituting (\ref{eq:ap8}) in (\ref{eq:ap9}) completes the proof as follows
\begin{align}
\pi_{+n}(m)= c \pi_{-n}(m)\sum_{j=1}^{Z}\left[\prod_{\hat a \in \mathbf{A}-\{a_n\}}\mathbb{I}(C_{a_n}'B_{j}\pi_{-n} <C_{\hat a}'B_{j}\pi_{-n} )\right]B_{mj}
\end{align}
\end{proof}
\subsection{Proof of Theorem \ref{prop:ideal}}\label{subsec:proofp2}
\begin{proof}The logarithm of the after-action public belief of learning problem (\ref{eq:setup2}) with benchmark information exchange Protocol~2, $\theta_n^{\rm full}$, is $\log\left(p(x|\Theta_n^{\rm full},G_n)\right)$. Recall that $\Theta_n^{\rm full}$ denotes the entire history of actions from previous nodes who have a path to node $n$ and $S_i$ denotes the set of all actions that $a_i$ depends on them. Also from definition of the transitive closure matrix (\ref{eq:transitivieclosurematrix}) and $t_n$ in (\ref{eq:tn}), the nodes who have a path to node $n$ are corrosponding to non-zero elements of $t_n$. Because if $t_n(i) = 1$, then there exists a path from node $i$ to node $n$. Therefore, the after-action public belief can be written as
\begin{align}\label{eq:ap1}
p(x|\Theta_n^{\rm full},a_n, G_n) =& c p(a_n|S_n,x) p (x | \{a_i ; a_i \in \Theta_{n}^{\rm full}\}) \nonumber\\=&  c\pi_0p(a_n|S_n,x) \prod_{a_i \in \Theta_n^{\rm full}} p(a_i| S_i,x).
\end{align}
Note that Bayes' theorem is used  recursively  to expand $ p (x | \{a_i ; a_i \in \Theta_{n}^{\rm full}\})$ and $S_i$ includes actions(from $\Theta_n^{\rm full}$) into account that $a_i$ depends on them. Taking the logarithm of both sides of (\ref{eq:ap1}) yields
\begin{align}
\theta_n^{\rm full}  =& \log\left(p(x|\Theta_n^{\rm full},a_n,G_n)\right)\nonumber\\
=& \log\left(c\pi_0 p(a_n|S_n,x) \prod_{a_i \in \Theta_n^{\rm full}} p(a_i| S_i,x)\right)\nonumber\\
=& \log\left(p(a_n|S_n,x)\right) + \sum_{t_n(i) \neq 0}\log\left(p(a_i| S_i,x)\right), \nonumber\\
=&  \sum_{i=1}^{n-1} t_n(i)\nu_i + \nu_n,
\end{align}
where $\nu_i$ denotes $\log\left(p(a_{i}|x,S_i)\right)$. Note that the normalizing constant $c$ and $\pi_0$ are omitted for the sake of simplicity as they are the same for both learning problems (\ref{eq:setup1}) with the constrained social learning Protocol~1 and (\ref{eq:setup2}) with the benchmark Protocol~2.
\end{proof}
\subsection{Proof of Theorem \ref{prop:mis}}\label{subsec:proofp3}
\begin{proof}
The aim here is to show that if $w_n = t_n\left(T_{n-1}'\right)^{-1}$ then $\hat \theta_n$ defined in (\ref{def:constraintestimate}) is exactly equal to $\theta^{\rm full}_n$ in (\ref{eq:deftheta}). Before proceeding, let us first rewrite (\ref{def:constraintestimate}) and (\ref{eq:deftheta}) using the following notations
\begin{align}\label{eq:new}
\theta_n^{\rm full} = \nu_n + (t_n \otimes {\mathbf{I}}_{d})\nu_{1:n-1},\nonumber\\
\hat \theta_n = \nu_n + (w_n \otimes {\mathbf{I}}_{d})\hat\theta_{1:n-1},
\end{align}
 where $\hat\theta_{1:n-1} \triangleq [\hat\theta_1',\ldots,\hat\theta_{n-1}']'$, $\nu_{1:n-1} \triangleq [\nu_1',\ldots,\nu_{n-1}']' \in \mathbb{R}^{(n-1)d \times 1} $. Here $\otimes$ denotes Kronecker (tensor) product and ${\mathbf{I}}_{d}$ denotes the $d\times d$ identity matrix.

 To prove Theorem~\ref{prop:mis}, we first start from \begin{equation}\label{eq:cri}\hat \theta_n = \theta^{\rm full}_n.\end{equation} Assume that (\ref{eq:cri}) holds for all $i$ where $1 \leq i \leq n$. From (\ref{def:constraintestimate}), $\theta^{\rm full}_n$ can be written as (given that Eq.~(\ref{eq:cri}) holds)
\begin{align}\label{eq:ap2}
\theta^{\rm full}_n = \hat\theta_n = (w_n \otimes \mathbf{I}_d)\hat\theta_{1:n-1} + \nu_n\nonumber\\
= (w_n \otimes \mathbf{I}_d)\theta^{\rm full}_{1:n-1} + \nu_n.
\end{align}
Eq. (\ref{eq:cri}) holds for all $i$ where $1 \leq i \leq n$. Therefore,  $\hat\theta_{1:n-1}=\theta^{\rm full}_{1:n-1}$. From (\ref{eq:deftheta}) in Theorem~\ref{prop:ideal},  $\theta_{1:n-1}^{\rm full}$ can be expressed as
\begin{equation}\label{eq:ap3}
\theta_{1:n-1}^{\rm full} = \left(T_{n-1}'\right)\nu_{1:n-1}.
\end{equation}
Note that in the derivation of (\ref{eq:ap3}), we use the definition of $t_{n-1}$ in (\ref{eq:tn}) as the first $n-2$ elements of $T_{n-1}$ and so on. Using (\ref{eq:ap3}), (\ref{eq:ap2}) can be written as
\begin{equation}\label{eq:ap4}
\theta^{\rm full}_n = (w_n \otimes \mathbf{I}_d)\left(T_{n-1}'\right)\nu_{1:n-1} +\nu_n.
\end{equation}
From (\ref{eq:deftheta}) in Theorem~\ref{prop:ideal}, we have another expression for $\theta_n^{\rm full}$. Comparing (\ref{eq:ap4}) and (\ref{eq:deftheta}) yields
\begin{align}\label{eq:ap5}
(t_n \otimes \mathbf{I}_d)\nu_{1:n-1} = (w_n \otimes \mathbf{I}_d)\left(T_{n-1}'\right)\nu_{1:n-1}\nonumber\\
=\left((w_nT'_{n-1})\otimes \mathbf{I}_d\right)\nu_{1:n-1}.
\end{align}
Note that in going from the first line to the second line in (\ref{eq:ap5}), the distributive property of tensor products is used. From (\ref{eq:ap5}) it can be inferred that $t_n = w_nT'_{n-1}$. As presented in Appendix~\ref{subsec:graphtheory}, $T_n$ is upper triangular matrix with ones in the diagonal. Therefore $T_n$ is invertible and $w_n = \left(t_nT_{n-1}'\right)^{-1}$. To complete the proof we need to start from  $w_n = \left(t_nT_{n-1}'\right)^{-1}$ and obtain $\hat \theta_n = \theta_n^{\rm full}$. This part of proof is straightforward and thus omitted from the paper. Note that the topological Constraint~1 says that if $b_n(j) =0$ then the  $j-$th entry of $\nu_{1:n-1}$ is not available to the node $n$ and thus the corresponding element of the weight vector $w_n(j)$ should be equal to zero as well. Also note that $\nu_n$ is computed by the network administrant and the data incest free public belief, $\pi_{-n}$, is available to the network administrator.
\end{proof}
\bibliographystyle{IEEEtran}
\bibliography{reff}
\end{document}